\documentclass[12pt]{article}

\usepackage{mathptmx}
\usepackage{times}
\usepackage{amsmath} 
\usepackage{amssymb} 
 \usepackage{amsthm} 
\usepackage{amsfonts} 
\usepackage{graphicx}
\usepackage{epsfig} 
\usepackage{algorithm}
\usepackage{algorithmic}
\usepackage{ifthen}
\usepackage{tikz}
\usetikzlibrary{trees}
\usepackage{subfig}
\usepackage{fullpage}

\newtheorem{definition}{Definition}
\newtheorem{theorem}{Theorem}
\newtheorem{lemma}{Lemma}

\newtheorem{example}{Example}

\newtheorem{corollary}{Corollary}

\newcommand{\compsepquery}{\textsc{CompRO}}
\newcommand{\compcotab}{\textsc{CompCoTable}}
\newcommand{\rowdecomp}{\textsc{RD}} 
\newcommand{\tabledecomp}{\textsc{TD}}
\newcommand{\var}{\texttt{Var}}
\newcommand{\lca}{\texttt{lca}}
\newcommand{\false}{\textsc{False}}
\newcommand{\true}{\textsc{True}}
\newcommand{\flag}{\textsc{Flag}}
\newcommand{\be}{\begin{enumerate}}
\newcommand{\ee}{\end{enumerate}}

\newcommand{\scream}[1]{} 
\newcommand{\newscream}[1]{} 
\newcommand{\ans}[1]{} 
\newcommand{\eat}[1]{}

\newcommand{\commentproof}[1]{}
\newcommand{\screamresolved}[1]{}

\newcommand{\tup}[1]{\mathbf{#1}}
\newcommand{\angb}[1]{\langle{#1}\rangle}

\begin{document}
%

\title{Faster Query Answering in Probabilistic Databases using Read-Once Functions\thanks{A shorter version of this paper
will appear in the proceedings of ICDT 2011.}}
\author{Sudeepa Roy\thanks{
University of Pennsylvania. Email: {\tt sudeepa@cis.upenn.edu}. Supported by NSF Award IIS-0803524.}
\and Vittorio Perduca\thanks{
University of Pennsylvania. Email: {\tt perduca@cis.upenn.edu}. Supported by NSF Award IIS-0629846.}
\and Val Tannen \thanks{
University of Pennsylvania. Email: {\tt val@cis.upenn.edu}. Supported in part by NSF Awards IIS-0629846 and IIS-0803524.
}}

%
%

\date{}
\maketitle

\begin{abstract}

A boolean expression is in {\em read-once form} if each of its
variables appears exactly once. When the variables denote independent
events in a probability space, the probability of the event denoted by
the whole expression in read-once form can be computed in polynomial
time (whereas the general problem for arbitrary expressions
is \#P-complete). Known approaches to checking read-once property
seem to require putting these expressions in disjunctive normal
form. 
\eat{
Unfortunately, it is known that for arbitrary
boolean formulas checking equivalence to a read-once form is unlikely
to be in PTIME (unless RP=NP). 
}
In this paper, we tell a better
story for a large subclass of boolean event expressions: those that are
generated by conjunctive queries without self-joins and on
tuple-independent probabilistic databases.

We first show that given a tuple-independent representation and the
{\em provenance graph} of an SPJ query plan without self-joins, we
can, without using the DNF of a result event expression, efficiently
compute its {\em co-occurrence graph}. From this, the read-once form
can already, if it exists, be computed efficiently using existing
techniques.  Our second and key contribution is a complete, efficient,
and simple to implement algorithm for computing the read-once forms
(whenever they exist) directly, using a new concept, that of {\em
co-table graph}, which can be significantly smaller than the
co-occurrence graph.

\end{abstract}

\newpage


\section{Introduction}
\label{sec:intro}
\scream{introduction.tex}

The computation of distributions for query answers on 
probabilistic databases is closely
related to the manipulation of boolean formulas. This connection has
led both to interesting theoretical questions and to implementation
opportunities. In this paper we consider the {\em tuple-independent} model
\footnote{This model has been considered as early as~\cite{CavalloP87}
  as well as in, eg.,~\cite{FuhrR97,Zimanyi97,GradelGH98,DalviS04}.}
for probabilistic databases. In such a model probabilistic databases are
represented by tables whose tuples $t$ are each annotated by a
probability value $p_t>0$, see Fig.~\ref{fig:example1}\subref{subfig:pdb}. Each tuple appears in a
possible world (instance) of the representation with probability $p_t$
independently of the other tuples. This defines a probability
distribution on all possible instances.  

\begin{figure}[t]
\centering 
\subfloat[]{\label{subfig:pdb}$
R=
\begin{array}{|c|c}
\cline{1-1}
a_1 & 0.3 \\
\cline{1-1}
b_1 & 0.4 \\
\cline{1-1}
a_2 & 0.6 \\
\cline{1-1}
\end{array}~~~~~
S=
\begin{array}{|cc|c}
\cline{1-2}
a_1 & c_1 & 0.1 \\
\cline{1-2}
b_1 & c_1 & 0.5 \\
\cline{1-2}
a_2 & c_2 & 0.2 \\
\cline{1-2}
a_2 & d_2 & 0.1 \\
\cline{1-2}
\end{array}~~~~~
T=
\begin{array}{|c|c}
\cline{1-1}
c_1 & 0.7 \\
\cline{1-1}
c_2 & 0.8 \\
\cline{1-1}
d_2 & 0.4 \\
\cline{1-1}
\end{array}
$
} \\
\subfloat[]{\label{subfig:tablevent}$R=
\begin{array}{|c|c}
\cline{1-1}
a_1 & w_1 \\
\cline{1-1}
b_1 & w_2 \\
\cline{1-1}
a_2 & w_3 \\
\cline{1-1}
\end{array}~~~~~
S=
\begin{array}{|cc|c}
\cline{1-2}
a_1 & c_1 & v_1 \\
\cline{1-2}
b_1 & c_1 & v_2 \\
\cline{1-2}
a_2 & c_2 & v_3 \\
\cline{1-2}
a_2 & d_2 & v_4 \\
\cline{1-2}
\end{array}~~~~~
T=
\begin{array}{|c|c}
\cline{1-1}
c_1 & u_1 \\
\cline{1-1}
c_2 & u_2 \\
\cline{1-1}
d_2 & u_3 \\
\cline{1-1}
\end{array}
$}\\
\subfloat[]{\label{subfig:query}$Q():- R(x),S(x,y),T(y)$}
\caption{\subref{subfig:pdb} A tuple-independent probabilistic database. \subref{subfig:tablevent} Event table representation. \subref{subfig:query} An \emph{unsafe} query. 
}  
\label{fig:example1}
\end{figure}

Manipulating all possible instances is impossibly unwieldy so
techniques have been
developed~\cite{LakshmananLRS97,FuhrR97,Zimanyi97} for obtaining the
query answers from the much smaller representation tables. This is
where boolean formulas make their entrance. The idea, by now
well-understood~\cite{DalviS04,BenjellounSHW06,DalviS07,AntovaJKO08},
is to define the relational algebra operators on tables whose tuples
are annotated with event expressions.  The event expressions are
boolean expressions whose variables annotate the tuples in the input
tables.  The computation of event expressions is the same as that used
in c-tables~\cite{ImielinskiL84} as models for incomplete and
probabilistic databases are closely related~\cite{GreenT06}.  Once the
event expressions are computed for the tuples in the representation
table of the query answer (which is in general {\em not}
tuple-independent), probabilities are computed according to the
standard laws.

The event expressions method was called {\em intensional}
``semantics'' by Fuhr and R{\"o}lleke and they observed
that with this method computing the query answer probabilities 
seems to require
exponentially many steps in general~\cite{FuhrR97}.
Indeed, the data complexity
\footnote{Here, and throughout the paper, the data
input consists of the representation tables~\cite{GradelGH98,DalviS04}
rather than the collection of possible worlds.}
of query evaluation on probabilistic
databases is \#P-complete, even for conjunctive 
queries~\cite{GradelGH98},
in fact even for quite simple boolean queries ~\cite{DalviS04} such as the query in Fig.~\ref{fig:example1}\subref{subfig:query}. 


But Fuhr and R{\"o}lleke also observe that certain event
independences can be taken advantage of, when present, to compute
answer probabilities in PTIME, with a procedure called 
{\em extensional} ``semantics''.  The idea behind the extensional
approach is the starting point for the remarkable
results~\cite{DalviS04,DalviS07a} of Dalvi and Suciu who discovered
that the conjunctive queries can be decidably and elegantly
separated into those whose data complexity is \#P-complete and those
for whom a {\em safe plan} taking the extensional approach can be
found.

Our starting point is the observation that even when the data
complexity of a query is \#P-complete (i.e. the query is \emph{unsafe}~\cite{DalviS04}), 
there may be classes of data
inputs for which the computation can be done with the extensional
approach, and is therefore in PTIME. We illustrate with a simple example.


\begin{example} 
\label{example1}
Consider the tuple-independent probabilistic database and the conjunctive query Q in Fig.~\ref{fig:example1}. Since the query Q is boolean
it has just one possible answer and the
event expression annotating it is
\footnote{To reduce the size of expressions
and following established tradition we use
$+$ for $\vee$ and $\cdot$ for $\wedge$,
and we even omit the latter in most terms.}
\begin{equation}
\label{equ:answer-exp}
f=w_1v_1u_1 + w_2v_2u_1 + w_3v_3u_2 + w_3v_4u_3
\end{equation}
This was obtained with the standard plan $\pi_{()}((R\bowtie S)\bowtie T)$.
However, it is equivalent to another boolean expression
\begin{equation}
\label{equ:answer-exp-sep}
(w_1v_1 + w_2v_2)u_1 + w_3(v_3u_2 + v_4u_3)
\end{equation}
which has the property that each variable
occurs exactly once.
\end{example}

Event (boolean) expressions in which each 
variable occurs exactly once are in {\em read-once form} 
(see~\cite{Newman92}).
For read-once forms the events denoted by non-overlapping subexpressions
are jointly independent, so we can use the key idea of the extensional
approach:

{\bf Fact~} If events $E_1,\ldots,E_n$ are jointly independent then
\begin{align}
P(E_1\cap\cdots\cap E_n) &~=~ P(E_1)\cdots P(E_n)  \label{eq:probinter}     \\
P(E_1\cup\cdots\cup E_n) &~=~ 1 - [1 - P(E_1)]\cdots [1 - P(E_n)]. \label{eq:probunion}
\end{align}

\begin{example}[Example \ref{example1} continued]
The probability of the answer~(\ref{equ:answer-exp-sep}) can be computed as follows
\begin{eqnarray*}
\lefteqn{P(f)=P((w_1v_1 + w_2v_2)u_1 + w_3(v_3u_2 + v_4u_3)) = } \\
 & & 1 - [1 - P(w_1v_1 + w_2v_2)P(u_1)]
    [1 - P(w_3)P(v_3u_2 + v_4u_3)]
\end{eqnarray*}
where
$$
P(w_1v_1 + w_2v_2) = 1 - [1- P(w_1)P(v_1)][1-P(w_2)P(v_2)]
$$
and
$$
P(v_3u_2 + v_4u_3) = 1 - [1-P(v_3)P(u_2)][1- P(v_4)P(u_3)].
$$
We can extend this example to an entire
class of representation tables of unbounded size. 
For each $n$, the relations $R,T$ will have $3n$ tuples 
while $S$ will have $4n$ tuples,
and the probability of the answer can be computed in time $O(n)$, see Appendix~\ref{sec:appexample}.

It is also clear that there is no relational algebra plan that 
directly yields {\rm (\ref{equ:answer-exp-sep})} above.
\end{example}

In fact, Fuhr and R{\"o}lleke (see~\cite{FuhrR97}, Thm.4.5) state that
probabilities can be computed by ``simple evaluation'' 
(i.e., by the extensional method)
if and only if the event expressions computed intensionally 
are in read-once form.
Moreover, the {\em safe plans} of~\cite{DalviS04} are such that all
the event expressions computed with the intensional method both on
intermediary relations and on the final answer, are in read-once form.

But more can be done.  Notice that the expression
(\ref{equ:answer-exp}) is not in read-once form but it is equivalent
to (\ref{equ:answer-exp-sep}) which is. Boolean expressions that are
equivalent to read-once forms have been called by various names, eg.,
separable, fanout-free~\cite{Hayes75}, repetition-free~\cite{Gurvich77},
$\mu$-expressions~\cite{Valiant84}, 
non-repeating~\cite{PeerP95},
but since the late 80's~\cite{HellersteinK89} the terminology seems to
have converged on {\em read-once}.  Of course, not all boolean
expressions are read-once, eg., $xy+yz+zx$ or $xy+yz+zu$ are not.

With this motivation we take the study of the following problem
as the goal of this paper: 

\medskip
\noindent
\textbf{Problem~} Given tuple-independent database $I$ and boolean
conjunctive query $Q$, when is $Q(I)$ read-once and if so, can its
read-once form be computed efficiently?
\eat{
\footnote{{\em We include this footnote in the submission version only!}
As it turned out, we obtained our results without being aware
of~\cite{SenDG10,GolumbicMR06,CorneilPS85,HellersteinK89,CorneilLB81} etc., 
in fact, of the whole literature on read-once and cographs. 
We found out about all these papers
very recently. We are not including this remark here as an excuse, 
but as an explanation why this submission might read somewhat as a
battle to emphasize those parts of our independent work that
distinguish it from, especially,~\cite{SenDG10}. We have, of course,
also modified our writing to fit the terminology of the read-once
and cograph literature.}
}
\medskip

It turns out that~\cite{GolumbicMR06} gives a fast algorithm
that takes a formula in irredundant disjunctive normal form, decides
whether it is read-once, and if it is, computes the read-once form
(which is in fact unique modulo associativity and commutativity).
The algorithm is based upon a characterization in terms of the formula's
co-occurrence graph given in~\cite{Gurvich91}.

Some terminology (taken up again in Section~\ref{sec:prelims}). 
Since we don't have anything to say about negation
or difference in queries we work only with {\em monotone} boolean
formulas (all literals are positive, only disjunction and conjunction
operations).  Disjunctive normal forms (DNFs) are disjunctions of {\em
  implicants}, which in turn are conjunctions of {\em distinct}
variables.  A {\em prime} implicant of a formula $E$ is one with a
minimal set of variables among all that can appear in DNFs equivalent
to $E$.  By absorption, we can retain only the prime implicants. The
result is called an {\em irredundant} DNF (IDNF) of $E$, and is unique
modulo associativity and commutativity.  The {\em co-occurrence graph}
of a boolean formula $E$ has its variables as nodes and has an edge
between $x$ and $y$ iff they both occur in the same prime implicant of
$E$.

For positive relational queries, the size of the IDNF of the boolean
event expressions is polynomial in the size of the table, but often
(and necessarily) exponential in the size of the query. This is a good 
reason for avoiding the explicit
computation of the IDNFs, and in particular for not relying on the algorithm 
in~\cite{GolumbicMR06}.  In recent and independent work 
Sen et al.~\cite{SenDG10} proved  that for the boolean expressions that
arise out of the evaluation of conjunctive queries without self-joins
the characterization in~\cite{Gurvich91} can be simplified and one
only needs to test whether the co-occurrence graph is a
``cograph''~\cite{CorneilLB81} which can be done
\footnote{Defining cographs seems unnecessary for this paper. 
It suffices to point out that most cograph recognition algorithms
produce (if it exists) something called a ``cotree'' (sigh) which in the 
case of co-occurrence graphs associated to boolean formulas is exactly
a read-once form!} in linear time~\cite{CorneilPS85}. 

It is also stated~\cite{SenDG10} that even for conjunctive queries 
without self-joins computing co-occurrence graphs likely requires obtaining  
the IDNF of the boolean expressions. One of our contributions in this paper 
is to show that an excursion through the IDNF is in fact not 
necessary because the co-occurrence graphs can be computed directly
from the {\em provenance graph}~\cite{KarvounarakisIT10,GreenKIT07}
that captures the computation of the query on a table. 
Provenance graphs are DAG representations of the event 
expressions in such a way that most common subexpressions for the entire table 
(rather than just each tuple) are not replicated. The smaller size of the 
provenance graphs likely provides practical speedups in the computations 
(compared for example with the provenance trees of~\cite{SenDG10}). 
Moreover, our approach
may be applicable to other kinds of queries, as long as their provenance graphs
satisfy a simple criterion that we identify.

To give more context to our results, we also note that 
Hellerstein and Karpinski\cite{HellersteinK89} have
shown that if $RP\neq NP$ then deciding whether an arbitrary
monotone boolean formula is read-once cannot be done in PTIME 
in the size of the formula.

The restriction to conjunctive queries without self-joins further
allows us to contribute improvements even over an approach that composes
our efficient computation of co-occurrence graphs with one of the
linear-time algorithms for cograph
recognition~\cite{CorneilPS85,HabibP05,BretscherCHP08}.
Indeed, we show that only a certain subgraph of the co-occurrence
graph (we call it the {\em co-table} graph) is relevant for our stated
problem. The co-table graph can be asymptotically
smaller than the co-occurrence graph for some classes of queries and instances.
To enable the use of only part of the co-occurrence graph
we contribute a novel algorithm that computes (when they exist) 
the read-once forms, using two new ideas: {\em row decomposition} and 
{\em table decomposition}. Using just connectivity tests (eg., DFS), our
algorithm is simpler to implement than the cograph
recognition algorithms in~\cite{CorneilPS85,HabibP05,BretscherCHP08}
and it has the potential of affecting the
implementation of probabilistic databases. 

Moreover, the proof of completeness for our algorithm does not use the
cograph characterization on which~\cite{SenDG10} relies. 
As such, the
algorithm itself provides an alternative new characterization of
read-once expressions generated by conjunctive queries without
self-joins. This may provide useful insights into extending the
approach to handle larger classes of queries.

Having rejected the use of co-occurrence graphs, Sen et
al.~\cite{SenDG10} provide a different approach that derives
efficiently the read-once form directly from the computations of the
trees underlying the boolean expressions, so called ``lineage trees'',
by merging read-once forms that correspond to
partial formulas. They provide a complexity analysis only for one
of the steps that their algorithm applies repeatedly.
However, to the best of our understanding of the asymptotic
complexity of their algorithm, it appears that our algorithm
is faster at least by a multiplicative factor of $k^2$ where 
$k$ is the number of tables, and the benefit can often
be more.

It is also important to note that neither the results of this paper, nor
those of~\cite{SenDG10} provide complexity dichotomies as does, 
eg.~\cite{DalviS04}. It is easy to 
give a family of probabilistic databases for which the query 
in Fig.~\ref{fig:example1}\subref{subfig:query} generates event expressions
of the following form:
$$
x_1x_2+x_2x_3+\cdots +x_{n-1}x_n+x_nx_{n+1}.
$$
These formulas are not read-once, but with a simple memoization (dynamic
programming) technique we can compute their probability in time linear in
$n$ (see Appendix~\ref{sec:discussions}).\\

\smallskip
\noindent
\textbf{Roadmap.~} 
In Section~\ref{sec:prelims} we review definitions, 
explain how to compute provenance DAGs for SPJ queries, and compare
the sizes of the co-occurrence and co-table graphs. Section~\ref{sec:co-graph}
presents a characterization of the co-occurrence graphs that
correspond to boolean expressions generated by conjunctive queries
without self-joins. The characterization uses the provenance DAG.
With this characterization we give an efficient algorithm for computing
the co-table (and co-occurrence) graph. In Section~\ref{sec:algo-query}
we give an efficient algorithm that, using the co-table graph, checks if the
result of the query is read-once, and if so computes its read-once form.
Putting together these two algorithms we obtain an efficient
query-answering algorithm that is {\em complete} for boolean conjunctive queries
without self-joins and tuple-independent databases that yield read-once
event expressions. In Section~\ref{sec:complexity} we compare the
time complexity of this algorithm with that of Sen et al.~\cite{SenDG10},
and other approaches that take advantage of past work in the read-once
and cograph literature. Related work, conclusions and ideas for 
further work ensue.
\section{Preliminaries}\label{sec:prelims}

A {\bf tuple-independent probabilistic database} is represented by
a usual (set-)relational database instance $I$
in which, additionally, every tuple
is annotated with a probability in $(0,1]$,
see for example Fig.~\ref{fig:example1}\subref{subfig:pdb}. 
We call this the {\em probability table representation}. We will denote by 
$\mathbf{R} = \{R_1,\ldots,R_k\}$ the relational schema of 
the representation.
\scream{I don't think we need to say anything
about attribute notation. In fact, why does the example ins sec 4 have attributes?}
By including/excluding each tuple independently with probability of  its 
annotation, the representation defines a set of $\mathbf{R}$-instances
called {\em possible worlds} and the obvious probability distribution on 
this set, hence a discrete probability space. For a given tuple $t\in R_i$
this space's event ``$t$ occurs'' (the set of possible worlds in 
which $t$ occurs) has probability exactly the annotation of $t$ in the 
representation.
Following the intensional approach~\cite{LakshmananLRS97,FuhrR97,Zimanyi97}
we also consider the {\em event table representation} which consists of
the same tables, but in which every tuple is annotated by its unique tuple id,
for example see Fig.~\ref{fig:example1}\subref{subfig:tablevent}. 

The {\bf tuple ids} play three distinct but related roles: (1) they identify
tuples uniquely over {\em all} tables and in fact we will often call
the tuple ids just tuples, (2) they are boolean variables, 
(3) they denote the events ``the tuple occurs'' in the
probability space of all possible worlds. The last two perspectives
can be combined by saying that the tuple ids are boolean-valued random
variables over said probability space. Moreover, an {\em event expression}
is a boolean expression with the tuple ids as variables.

The intensional approach further defines the semantics of 
the {\bf relational algebra operators}
on event tables, i.e., relational instances
in which the tuples are annotated with event expressions. In this paper
we will only need monotone boolean expressions because our queries only use
selection, projection and join and these operators do not introduce negation.
Otherwise, joins produce conjunctions, projections produce disjunctions,
and selections erase the non-compliant tuples. It is worth observing that
the relational algebra on event tables is essentially a particular case
of the algebra on c-tables~\cite{ImielinskiL84}, and precisely a particular
case of the relational algebra on semiring-annotated relations~\cite{GreenKT07}.

Since by now they are well understood (see the many papers we cited so
far), we do not repeat here the definition of select, project, and
join on tables annotated with boolean event expressions but instead we
explain how they produce {\bf provenance graphs}.  The concept that we
define here is a small variation on the provenance graphs defined
in~\cite{KarvounarakisIT10,GreenKIT07} where conjunctive queries (part
of mapping specifications) are treated as a black box. It is important
for the provenance graphs used in this paper to reflect the structure
of different SPJ query plans that compute the same conjunctive query.

A provenance graph (PG) is a directed acyclic graph (DAG) $H$ such that the
nodes $V(H)$ of $H$ are labeled by variables or by the
operation symbols $\cdot$ and $+$. As we show below, each node
corresponds to a tuple in an event table that represents the
set of possible worlds of either the input database or some
intermediate database computed by the query plan.
An edge $u\rightarrow v$ is in $E(H)$ if the tuple corresponding
to $u$ is computed using the tuple corresponding to $v$ in either
a join (in which case $u$ is labeled with $\cdot$) or a projection
(in which case $u$ is labeled with $+$).
The nodes with no outgoing edges
are those labeled with variables and are called leaves while the
nodes with no incoming edges are called roots (and can be labeled with
either operation symbol). Provenance
graphs (PGs) are closely related to the {\em lineage trees} of~\cite{SenDG10}.
In fact, the lineage trees are tree representations of the boolean event
expressions, while PGs are more economical:
they represent the same expressions but without the multiplicity
of common subexpressions. Thus, they are associated with an entire table
rather than with each tuple separately, each root of the graph
corresponding to a tuple in the table. 
\footnote{Note that to facilitate the comparison
with the lineage trees
the edge direction here is the opposite of the direction 
in~\cite{KarvounarakisIT10,GreenKIT07}.}

\begin{figure}[htbp]
\begin{center}
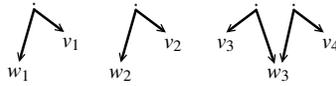
\caption{Provenance graph for $R\bowtie S$.}
\label{fig:pg-partial}
\end{center}
\end{figure}

\begin{figure}[htbp]
\begin{center}
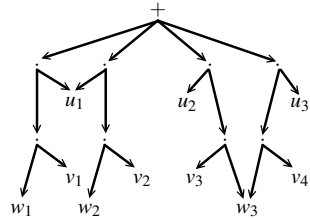
\caption{Provenance graph for $\pi_{()}((R\bowtie S)\bowtie T)$.}
\label{fig:pg}
\end{center}
\end{figure}

We explain how the SPJ algebra works on tables with PGs.
If tables $R_1$ and $R_2$ have PGs $H_1$
and $H_2$ then the PG for $R_1\bowtie R_2$ is
constructed
as follows. Take the disjoint union $H$ of $H_1$ and $H_2$. For every
$t_1\in R_1$ and $t_2\in R_2$ that do join, 
add a new root labeled with $\cdot$ and make
the root of $H_1$ corresponding to $t_1$ and that of $H_2$ corresponding
to $t_2$ children of this new root. Afterwards, delete (recursively)
any remaining roots from $H_1$ and $H_2$. For example, referring again
to Fig.~\ref{fig:example1}, the PG associated with the table computed
by $R\bowtie S$ is shown in Fig.~\ref{fig:pg-partial}.

For selection, delete (recursively) the roots that correspond
to the tuples that do not satisfy the selection predicate.
For projection, consider a table $T$ with PG $H$ and 
$X$ a subset of its attributes. The PG for $\pi_X R$
is constructed as follows. For each $t\in\pi_X R$, let
$t_1,\ldots,t_m$ be all the tuples in $R$ that $X$-project to $t$.
Add to $H$ a new root labeled with $+$ and make the roots in $H$ corresponding
to $t_1,\ldots,t_m$ the children of this new root. Referring again
to Fig.~\ref{fig:example1}, the PG associated with the result
of the query plan $\pi_{()}((R\bowtie S)\bowtie T)$
is shown in Fig.~\ref{fig:pg}. Since the query
is boolean, this PG has just one root.

The boolean event expressions that annotate tuples in the event tables
built in the intensional approach can be read off the provenance
graphs. Indeed, if $t$ occurs in an (initial, intermediate, or final)
table $T$ whose PG is $H$, then, starting at the root $u$ of $H$
corresponding to $t$, traverse the subgraph induced by all the nodes
reachable from $u$ and build the boolean expression recursively
using parent labels as operation symbols and the subexpressions
corresponding to the children as operands. 
For example, we read 
$(((w_1\cdot v_1)\cdot u_1)+
(u_1\cdot (w_2\cdot v_2))+
(u_2\cdot (v_3\cdot w_3))+
((w_3\cdot v_4)\cdot u_3))$
off the PG in Fig.~\ref{fig:pg}.

The focus of this paper is the case when the boolean (event) expressions
are {\bf read-once}, i.e., they are equivalent to expressions in which
every variables occurs exactly once, the latter said to be in
{\em read-once form}. For boolean expressions that are read-once, 
the read-once form is unique (modulo associativity and commutativity)
\footnote{This seems to have been known for a long time. We could not
find an explicitly stated theorem to this effect in the literature, 
but, for example, it is clear that the result of the algorithm 
in~\cite{GolumbicMR06} is uniquely determined by the input.}.
The interest in read-once formulas derives from the fact that in tuple-independent databases
the tuples in the input representation occur {\em independently} in
possible worlds. More complex boolean expressions denote
events whose probability needs to be computed from the probabilities
of the variables, i.e., the probabilities of the independent ``tuple occurs''
events in the input. When such event expressions are in read-once form
their probability can be computed efficiently in linear time in the 
number of variables using the rules~(\ref{eq:probinter}) and~(\ref{eq:probunion})
in Section~\ref{sec:intro}.

\eat{
As observed in the introduction, the probability of a boolean expression in read-once form 
P(\sum_i f_i) & =1-\prod_i(1-P(f_i))\\
P(\prod_i f_i) & = \prod_i P(f_i), 
\end{align*}
where the $f_i$s are boolean expressions in read-once form not sharing variables, $\cap_i\var(f_i)=\phi$.
}

Given a tuple-independent probabilistic database and an SPJ query
plan, hence the resulting provenance graph, our objective in this
paper is to decide efficiently when the boolean expression(s) read off
the PG are read-once, and when they are, to compute their read-once
form(s) efficiently, hence the associated probability(es).

In this paper we consider only {\bf boolean
conjunctive queries}. We can do this
without loss of generality because we can associate to a non-boolean
conjunctive query $Q$ and an instance $I$ 
a set of boolean queries in the usual manner:
for each tuple $t$ in the answer relation $Q(I)$,
consider the boolean conjunctive query $Q_t$ which
is obtained from $Q$ by replacing the head variables
with the corresponding values in $t$.
Note that the PGs that result from boolean queries have 
{\em exactly one root}. We will also use $Q(I)$
to denote the boolean expression generated
by evaluating the query $Q$ on instance $I$,
which may have different (but equivalent) forms
based on the query plan.

Moreover, we consider only queries {\bf without self-join}.
Therefore our queries have the form
$$
Q() :- R_1(\mathbf{x}_1), \ldots , R_k(\mathbf{x}_k)
$$
where $R_1,\ldots,R_k$ are all {\em distinct} table names
while the $\mathbf{x}_i$'s are tuples of FO variables
\footnote{FO (first-order) is to emphasize the distinction
between the variables in the query subgoals and the variables
in the boolean expressions.} or constants,
possibly with repetitions, matching the arities of the tables.
\scream{BTW, along the way we must have understood also 
a characterization of boolean formulas that result from such queries
at least we know what their IDNF looks like.}
If the database has tables that do not appear in the query, they
are of no interest, so we will always assume that our queries
feature all the table names in the database schema $\mathbf{R}$.

As we have stated above, we only need to work with {\em monotone}
boolean formulas (all literals are positive, only disjunction and
conjunction operations).  Every such formula is equivalent to (many)
{\bf disjunctive normal forms} (DNFs) which are disjunctions of
conjunctions of variables.  These conjunctions are called {\em
  implicants} for the DNF.  By idempotence we can take the variables
in an implicant to be distinct and the implicants of a DNF to be
distinct from each other.  A {\em prime} implicant of a formula $f$ is
one with a minimal set of variables among all that can appear in DNFs
equivalent to $f$.  By absorption, we can retain only the prime
implicants in a DNF. The result is called {\em the irredundant} DNF
(IDNF) of $f$, as it is uniquely determined by $f$ (modulo associativity
and commutativity). We usually denote it by $f_{IDNF}$.
Note that in particular the set of prime implicants
is uniquely determined by $f$.

The {\bf co-occurrence graph}, notation $G_{co}$, of a boolean formula
$f$ is an undirected graph whose set of vertices $V(G_{co})$ is the
set $\var(f)$ of variables of $f$ and whose set $E(G_{co})$ of edges is
defined as follows: there is an edge between $x$ and $y$ iff they both
occur in the same prime implicant of $f$. Therefore, $G_{co}$ is
uniquely determined by $f$ and it can be constructed from
$f_{IDNF}$. This construction is quadratic \scream{can be improved?}
in the size of $f_{IDNF}$ but of course $f_{IDNF}$ can be
exponentially larger than $f$.  Fig. \ref{fig:cograph} shows the
co-occurrence graph for the boolean expression $f$ in equation 
(\ref{equ:answer-exp}) of Example~\ref{example1}.
As this figure shows, the co-occurrence graphs for expressions generated
by conjunctive queries without self join are always 
\emph{$k$-partite}\footnote{A graph
$(V_1 \cup \cdots \cup V_k, E)$ is
\emph{k-partite}, if for any edge $(u, v) \in E$ where $u \in V_i$
and $v \in V_j$, $i \neq j$.} graphs on tuple variables from $k$ different
tables.

\begin{figure}[htbp]
\begin{center}
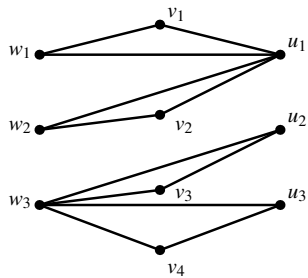
\caption{$G_{co}$ for $f$ in Example \ref{example1}.}
\label{fig:cograph}
\end{center}
\end{figure}

We are interested in the co-occurrence graph $G_{co}$ of a boolean formula $f$
because it plays a crucial role in $f$ being read-once. Indeed 
\cite{Gurvich91}~has shown that a monotone $f$ is read-once iff
(1) it is ``normal'' and (2) its $G_{co}$ is a ``cograph''. We don't need
to discuss normality because \cite{SenDG10}~has shown that for formulas
that arise from conjunctive queries without self-joins it follows from
the cograph property. We will also avoid defining what a cograph is
(see~\cite{CorneilLB81,CorneilPS85}) except to note that
cograph recognition can be done in linear 
time~\cite{CorneilPS85,HabibP05,BretscherCHP08} 
and that when applied to the co-occurrence
graph of $f$ the recognition algorithms also produce, in effect, the
read-once form of $f$, when it exists.

Although the co-occurrence graph of $f$ is defined in terms of
$f_{IDNF}$, we show in Section~\ref{sec:co-graph} that when $f$ is
the event expression produced by a boolean conjunctive query
without self-joins then we can efficiently compute the $G_{co}$
of $f$ from the provenance graph $H$ of any plan for the query.
\scream{discuss that we use ANY plan in related work}
Combining this with any of the cograph recognition algorithms
we just cited, this yields one algorithm for the goal of our paper,
which we will call a {\bf cograph-help algorithm}.

Because it uses the more general-purpose step of cograph recognition
a cograph-help algorithm will not fully take advantage of the restriction
to conjunctive queries without self-joins. Intuitively, with this restriction
there may be lots of edges in $G_{co}$ that are irrelevant because they link
tuples that are not joined by the query. This leads us to the notion of
co-table graph defined below.

Toward the definition of the co-table graph we also need that
of {\bf table-adjacency graph}, notation $G_T$. 
Given a boolean query without self-joins 
$Q() :- R_1(\mathbf{x}_1), \ldots , R_k(\mathbf{x}_k)$
the vertex set $V(G_T)$ is the set of $k$ table names $R_1, \cdots, R_k$.
We will say that $R_i$ and $R_j$ are \emph{adjacent} iff 
$\mathbf{x}_i$ and $\mathbf{x}_j$ have at least one FO variable in common
i.e., $R_i$ and $R_j$ are joined by the query.
The set of edges $E(G_T)$ consists of the pairs of adjacent table names.
The table-adjacency graph $G_T$ for the query in
Example \ref{example1} is depicted in Fig.~\ref{fig:tadj}.
\scream{Define $m_T$ in sec 4 where it might be used}

\begin{figure}[t]
\begin{center}
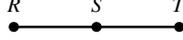
\caption{$G_T$ for the relations in Example \ref{example1}.}
\label{fig:tadj}
\end{center}
\end{figure}  

\par The table-adjacency graph $G_T$ helps us remove
edges irrelevant to a query from the graph $G_{co}$.  For example, if there
is an edge between $x\in R_i$ and $x'\in R_j$ in $G_{co}$, but
there is no edge between $R_i$ and $R_j$ in $G_T$, then (i) either
there is no path connecting $R_i$ to $R_j$ in $G_{T}$ (so all tuples
in $R_i$ pair with all tuples in $R_j$), or, (ii) $x$ and
$x'$ are connected in $G_{co}$ via a set of tuples $x_1, \cdots,
x_\ell$, such that the tables containing these tuples are connected by
a path in $G_T$. Our algorithm in Section~\ref{sec:algo-query} shows
that all such edges $(x, x')$ can be safely deleted from $G_{co}$
for the evaluation of the query that yielded $G_T$.

\begin{definition}
The {\bf co-table graph} $G_C$ is the subgraph of $G_{co}$ with 
$V(G_{C})=V(G_{co})$ and such that given two tuples $x\in R_i$ and 
$x'\in R_j$ there is an edge $(x,x')\in E(G_{C})$ iff $(x,x')\in E(G_{co})$
and $R_i$ and $R_j$ are adjacent in $G_T$. 
\end{definition}

The co-table graph $G_C$ generated by the event tables
and query in Fig.~\ref{fig:example1} is shown in Fig.~\ref{fig:cotable}
(it is {\em not} a cograph!).

\begin{figure}[t]
\begin{center}
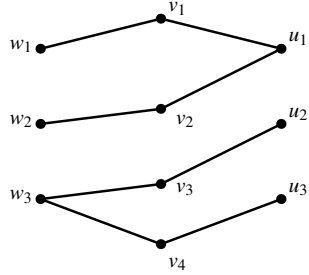
\caption{$G_C$ for $f$ in Example~\ref{example1}.}
\label{fig:cotable}
\end{center}
\end{figure}

{\bf Co-occurrence graph vs. co-table graph}~
The advantage of using the co-table graph 
instead of the co-occurrence graph is most dramatic
in the following example:

\begin{example}
\label{example2}
Consider $Q() :- R_1(\mathbf{x}_1),R_2(\mathbf{x}_2)$
where $\mathbf{x}_1$ and $\mathbf{x}_2$ have no common FO variable.
Assuming that each of the tables $R_1$ and $R_2$
has $n$ tuples, $G_{co}$ has $n^2$ edges 
while $G_C$ has none. A cograph-help algorithm must spend $\Omega(n^2)$ time
even if it only reads $G_{co}$!
\end{example}

On the other hand, $G_C$ can be as big as $G_{co}$. In fact, when
$G_T$ is a complete graph (see next example), $G_C=G_{co}$.

\begin{example}
\label{example3}
Consider $Q() :- R_1(\mathbf{x}_1,y),\ldots,R_k(\mathbf{x}_k,y)$
where $\mathbf{x}_i$ and $\mathbf{x}_j$ have no common FO variable
if $i\neq j$. Here $G_T$ is the complete graph on $R_1,\ldots,R_k$ and
$G_C=G_{co}$. 
\end{example}
However, it can be verified that both our algorithm and the cograph-help algorithm
have the same time complexity on the above example.

\section{Computing the Co-Table Graph}\label{sec:co-graph}
In this section we show that given as input 
the provenance DAG $H$ of a boolean conjunctive query {\em plan} 
without self-joins $Q$ on a table-independent database representation $I$,
the co-table graph $G_C$ and the co-occurrence graph $G_{co}$ of the
boolean formula $Q(I)$ (see definitions in section~\ref{sec:prelims}) 
can be computed in poly-time in the sizes of $H,I$ and $Q$.
\par
It turns out that $G_C$ and $G_{co}$ are computed
by similar algorithms, one being a minor modification
of the other.
As discussed in section~\ref{sec:intro}, 
the co-occurrence graph $G_{co}$ can then be used in conjunction
with cograph recognition algorithms 
(eg.,~\cite{CorneilPS85, HabibP05,BretscherCHP08}), 
to find the read-once form of $Q(I)$ if it exists.
On the other hand, the smaller co-table graph $G_C$ is used
by our algorithm described in section~\ref{sec:algo-query}
for the same purpose.

\screamresolved{moved to intro
\scream{Maybe this should go in the intro} 
It is important to note
that our algorithms depend essentially on the fact that the boolean
expressions are generated by plans for conjunctive queries (SPJ
plans) that does not have any self-joins. It is not hard to see that adding union generates arbitrary
monotone boolean formulas. Hellerstein and
Karpinski~\cite{HellersteinK89} have shown that if $RP\neq NP$
then deciding whether an arbitrary \screamresolved{\scream{monotone??} not mentioned}
boolean formula is read-once cannot be done in poly-time.
\footnote{The algorithm in~\cite{GolumbicMR06} assumes, of course,
that the input is in IDNF.} Therefore, it is unlikely that 
a result similar to ours holds for unions of conjunctive queries.
Even it will be interesting to see whether our approach extends to
conjunctive queries that allow self-joins.
}

We use $\var({f})$ to denote the sets of variables in a monotone boolean expression $f$. 
Recall that the provenance DAG $H$ is a layered graph
where every layer corresponds to a select, project or join operation in the query plan.
We define the \emph{width} of $H$ as the maximum number of nodes at any layer
of the DAG $H$ and denote it by $\beta_H$.
The main result in this section is summarized by the following theorem.

\begin{theorem}\label{thm:lca}
Let $f=Q(I)$ be the boolean expression computed by the query plan $Q$
on the table representation $I$ ($f$ can also be read off the provenance
graph of $Q$ on $I$, $H$), $n = |\var(f)|$ be the number of variables in $f$,
$m_H = |E(H)|$ be the number of edges of $H$, $\beta_H$ be the width of 
$H$, and $m_{co} = |E(G_{co})|$ 
be the number of edges of $G_{co}$, the co-occurrence graph of $f$.
\begin{enumerate}
	\item 
$G_{co}$ can be computed in time $O(n m_H + \beta_H m_{co})$.
\screamresolved{\scream{This $\beta$ needs more discussion}}
       \item
Further, the co-table graph $G_C$ of $f$
can be computed in time
$O(n m_H + \beta_H m_{co} + k^2\alpha \log \alpha)$
where 
$k$ is the number of tables in $Q$,
and $\alpha$ is the maximum arity (width) of the tables in $Q$.
\end{enumerate}
\end{theorem}

\subsection{LCA-Based Characterization of the Co- Occurrence Graph}\label{sec:lca_charac}
Here we give a characterization of the presence of an edge $(x, y)$
in $G_{co}$ based on the \emph{least common ancestors} of $x$ and $y$
in the graph $H$.

Again, let $f=Q(I)$ be the boolean expression computed by the query plan $Q$
on the table representation $I$. As explained in section~\ref{sec:prelims}
$f$ can also be read off the provenance graph $H$ of $Q$ and $I$ since
$H$ is the representation of $f$ without duplication of common subexpressions.

The absence of self-joins in $Q$ implies the following.

\begin{lemma}\label{lem:dnf_irr}
The DNF generated by expanding $f$ (or $H$)
using only the distributivity
rule is in fact the IDNF of $f$ up to idempotency (i.e. repetition of the
same prime implicant is allowed). 
\end{lemma}
\begin{proof}
Let $g$ be the DNF generated from $f$ by applying distributivity
repeatedly. Due to the absence of self-joins $g$ every implicant
in $g$ will have exactly one tuple from every table.
Therefore, for any two implicants in $g$ the set of variables in one is not
a strict subset of the set of variables in the other and further
absorption (eg., $xy + xyz = xy$) does not apply. (At worst, two implicants
can be the same and the idempotence rule reduces one.)
Therefore, $g$ is also irredundant and hence {\em the} IDNF of $f$
(up to commutativity and associativity).\end{proof}

Denote by $f_{IDNF}$ the IDNF of $f$, which, as we have seen, 
can be computed from $f$ just by applying distributivity.
%
%

As with any DAG, we can talk about the nodes of $H$ in terms of
successors, predecessors, ancestors, and descendants, and finally
about the {\em least common ancestors} of two nodes, denoted 
$\lca(x,y)$. Because $H$ has a root $\lca(x,y)$ is never empty.
When $H$ is a tree, $\lca(x,y)$ consists of a single node.
For a node $u \in V(H)$, we denote the set of leaf
variables which are descendants of $u$ by $\var(u)$ 
(overloaded notation warning!);
in other words, a variable $x$ belongs to $\var(u)$, 
$u \in V(H)$, if and only if $x$ is reachable
from $u$ in $H$. 
Now we prove the key lemma of this section:

\begin{lemma}\label{lem:dnf-lca}
Two variables $x, y \in \var(f)$ belong together to a (prime) implicant
of $f_{IDNF}$ if and only if the set $\lca(x, y)$ contains a $\cdot$-node.
\end{lemma}
\begin{proof}
(if) Suppose $\lca(x, y)$ contains a $\cdot$-node $u$, i.e., $x, y$ are 
both descendants of two distinct successors $v_1, v_2$ of $u$. 
Since the $\cdot$
operation multiplies all variables in $\var(v_1)$ with all variables 
in $\var(v_2)$,
$x$ and $y$ will appear together in some implicant in $f_{IDNF}$ which will not
be absorbed by other implicants by Lemma~\ref{lem:dnf_irr}.
\par
(only if) Suppose that $x, y$ appear together in an implicant of
$f_{IDNF}$ and $\lca(x, y)$ contains no $\cdot$-node. 
Then no $\cdot$-node in $V(H)$ has $x, y$ in $\var(v_1), \var(v_2)$, 
where $v_1, v_2$ are its two distinct successors (note that any $\cdot$-node
in a provenance DAG $H$ can have exactly two successors).
\screamresolved{\scream{How many distinct successors does a node have?}}
This implies that $x$ and $y$ can never be multiplied, contradiction.
\end{proof}

Since there are exactly $k$ tables in the query plan, 
every implicant in $f_{IDNF}$ will
be of size $k$. Therefore:

\begin{lemma}\label{lem:unique_successor}
For every variable $x \in \var(f)$ and
$\cdot$-node $u \in V(H)$, if $x \in \var(u)$,
then $x \in \var(v)$ for exactly one successor $v$ of $u$.
\end{lemma}
\begin{proof}
If $x \in \var(f)$ belongs to 
$\var(v_1), \var(v_2)$ for two distinct successors 
$v_1, v_2$ of $u$, then some implicant in $f_{IDNF}$ will have $< k$ variables
since $x\cdot x = x$ by idempotence. 
\end{proof}

The statement of Lemma~\ref{lem:dnf-lca} provides a criterion
for computing $G_{co}$ using the computation of least common ancestors
in the provenance graph, which is in often more efficient than computing 
the entire IDNF. We have shown that this criterion is satisfied in the case
of conjunctive queries without self-joins. But it may also be satisfied
by other kinds of queries, which opens a path to identifying other 
cases in which such an approach would work.

\subsection{Computing the Table-Adjacency Graph}\label{sec:comp_g_t}

\noindent
It is easier to describe the computation of $G_T$ if we use the
query in rule form $Q() :- R_1(\mathbf{x}_1), \ldots , R_k(\mathbf{x}_k)$.

The rule form can be computed in linear time from the SPJ query plan.
Now the vertex set $V(G_T)$ is the set of table names $R_1, \cdots, R_k$.
and an edge exists between $R_i, R_j$ iff $\mathbf{x}_i$
and $\mathbf{x}_j$ have at least one FO variable in common
i.e., $R_i$ and $R_j$ are joined. Whether or not
such an edge should be added can be decided in time
$O(\alpha \log\alpha)$ by sorting and intersecting
$\mathbf{x}_i$ and $\mathbf{x}_j$. Here $\alpha$ is the maximum
arity (width) of the tables $R_1, \cdots, R_k$.
Hence $G_T$ can be computed in time $O(k^2\alpha \log \alpha)$.

\subsection{Computing the Co-Table Graph}

Recall that co-table graph $G_C$ is a subgraph of the co-occurrence 
graph $G_{co}$
where we add an edge between two variables $x, y$, only if the tables containing
these two tuples are adjacent in the table-adjacency graph $G_T$.
Algorithm~\ref{alg:lca}  \compcotab\ 
constructs the co-table graph $G_C$ 
by a single \emph{bottom-up} pass over the graph $H$.

\begin{algorithm}[h!t!]
\caption{\em Algorithm \compcotab}
{\bf Input: Query plan DAG $H$ and table-adjacency graph $G_T$} \\
{\bf Output: Co-table graph $G_C$.~~~~~~~~~~~~~~~~~~~~~~~~~~~~}
\begin{algorithmic}[1] \label{alg:lca}
	\STATE{-- Initialize $V(G_C) = \var(f)$, $E(G_C) = \phi$.}
	\STATE{-- For all variables $x \in \var(f)$, set $\var(x) = \{x\}$.}
	\STATE {-- Do a \emph{topological sort} on $H$ and \emph{reverse} the sorted order.}
	\FOR{every node $u \in V(H)$ in this order}
		\STATE{\textit{/* Update $\var(u)$ set for both $+$-node and $\cdot$-node $u$*/}}
		\STATE{-- Set $\var(u) = \bigcup_{v} \var(v)$, where the union is
		over all successors $v$ of $u$.}
		\IF{$u \in V(H)$ is a $\cdot$-node}
			\STATE{\textit{/* Add edges to $G_C$ only for a $\cdot$-node*/}}
			\STATE{-- Let $v_1, v_2$ be its two successors.}
			\FOR{every two variables $x \in \var(v_1)$ and $y \in \var(v_2)$}\label{step:check}
			 	\IF{(i) the tables containing $x, y$ are adjacent in $G_T$ and 
			 	(ii) the edge $(x, y)$ does not exist in $E(G_C)$ yet}\label{step:table}
			 		\STATE{-- Add an edge between $x$ and $y$ in $E(G_C)$.}		
				\ENDIF
			\ENDFOR
		\ENDIF
	\ENDFOR
\end{algorithmic}
\end{algorithm}

It is easy to see that a minor modification of 
the same algorithm can be used
to compute the co-occurrence graph $G_{co}$: in Step~\ref{step:table}
we simply skip the check whether the tables
containing the two tuples are adjacent in $G_T$.
Since this is the only place where $G_T$ is used, the time
for the computation of $G_C$ does not include the time related
to computing/checking $G_T$.

\paragraph{Correctness.~~}
By a simple induction, it can be shown that the set $\var(u)$ is correctly computed at every step,
i.e., it contains the set of all nodes which are reachable from $u$ in $H$
(since the nodes are processed in reverse topological order and $\var(u)$ is union of 
$\var(v)$ for over all successors $v$ of $u$). 
Next lemma shows that algorithm \compcotab\ correctly builds the co-table graph $G_C$ (proof in Appendix~\ref{app:proof_sec3}).

\begin{lemma}
\label{lemma:corcompcotable}
Algorithm \compcotab\ adds an edge $(x, y)$ to $G_{C}$ if and only if $x, y$
together appear in some implicant in $f_{IDNF}$ and the tables containing $x, y$
are adjacent in $G_T$.
\end{lemma}

\paragraph{Time Complexity.~~} 
Here we give a sketch of the time complexity analysis, details can be found
in the appendix (Section~\ref{sec:app_time_lca}).
Computation of the table adjacency graph takes $O(k^2 \alpha \log \alpha)$ time
as shown in Section~\ref{sec:comp_g_t}. 
The total time complexity of algorithm \compcotab\ as given in Theorem~\ref{thm:lca}
is mainly due to two operations: (i) computation of the $\var(u)$ set at every internal
node $u \in V(H)$, and (ii) to perform the test for pairs $x, y$ at
 two distinct children of a $\cdot$-node, whether the edge $(x, y)$ already exists
in $G_C$, and if not, to add the edge. 
\par
We show that the total time needed
for the first operation is $O(nm_H)$ in total: for every internal node
$u \in V(H)$ we can scan the variables sets of all its immediate successor in $O(nd_u)$ time to compute
$\var(u)$, where $d_u$ is the outdegree of node $u$ in $H$. This gives total $O(nm_H)$ time.
On the other hand, for adding edges $(x, y)$ in $G_C$,
it takes total $O(m_{co}\beta_H)$ time during the execution of the algorithm, where $m_{co}$
is the number of edges in the co-occurrence graph (and not in the co-table graph, even if 
we compute the co-table graph $G_C$) and $\beta_H$ is the width of the graph $H$. 
To show this, we show that two variables $x, y$ are considered by the algorithm at Step~\ref{step:check}
if and only if the edge $(x, y)$ already exists in the co-occurrence graph $G_{co}$, however, the edge
may not be added to the co-table graph $G_C$ if the corresponding tables are not adjacent in the table adjacency
graph $G_T$. We also show that any such edge $(x, y)$ will be considered 
at a unique level of the DAG $H$.
In addition to these operations, the algorithm
does initialization and a topological sort on the vertices which take
$O(m_H + n_H)$ time ($n_H = |V(H)|$) and are dominated by the these two operations.
\section{Computing the Read-Once Form} \label{sec:algo-query}
Our algorithm \compsepquery\ (for \emph{Compute Read-Once}) takes an instance $I$ of the schema $\tup{R} = R_1, \cdots, R_k$, 
a query $Q() :- R_1(\mathbf{x_1}), R_2(\mathbf{x_2}), \cdots, R_k(\mathbf{x_k})$ along with the
table adjacency graph $G_T$ and co-table graph $G_{C}$
computed in the previous section as input, and outputs whether $Q(I)$ is read-once. 
(if so it computes its unique read-once form). 

\eat{
We describe the algorithm in Section~\ref{sec:algo_query}, 
and discuss the correctness of the algorithm in Section~\ref{sec:query_correctness}.
In section~\ref{sec:query_time}, we give a time complexity analysis of this algorithm 
in detail. Finally, we illustrate our algorithm using an example in Section~\ref{sec:algo_eg}.
The following theorem summarizes our results in this section:
}

\begin{theorem}\label{thm:query_time}
Suppose we are given a query $Q$, a table-independent database representation $I$, the
co-table graph $G_C$ and the table-adjacency graph $G_T$ for $Q$ on $I$
as inputs.  Then
\begin{enumerate}
	\item Algorithm \compsepquery\ 
decides correctly whether the 
expression generated by evaluating $Q$ on $I$ is read-once, and if yes, it returns the
unique read-once form of the expression, and,
 \item Algorithm \compsepquery\ runs in time
$O(m_T\alpha \log \alpha + (m_C+n) \min(k, \sqrt{n}))$, 
\end{enumerate}
where $m_T = |E(G_T)|$ is the number of edges in $G_T$,
$m_{C} = |E(G_C)|$ is the number of edges in $G_C$, 
$n$ is the total number of tuples in $I$,
$k$ is the number of tables, and $\alpha$ is the maximum size of any subgoal.
\end{theorem}
\subsection{Algorithm \compsepquery}\label{sec:algo_query}

In addition to the probabilistic database with tables $R_1, \cdots, R_k$
and input query $Q$, our algorithm also takes the  \emph{table-adjacency graph $G_T$}
and the \emph{co-table graph $G_C$} computed in the first phase as discussed in
Section~\ref{sec:co-graph}. 
The co-table graph $G_C$ also helps us to remove \emph{unused tuples} from 
all the tables which do not appear in the final expression -- every unused tuple
won't have a corresponding node in $G_C$.
So from now on we can assume wlog. that every tuple in every table 
appears in the final expression $f$. 
\par
The algorithm $\compsepquery$ uses two decomposition operations: \emph{Row decomposition}
is a \emph{horizontal} decomposition operation which partitions the rows or tuples in 
every table into the same number of groups and forms a set of sub-tables from every table. 
On the other hand, \emph{Table decomposition} is a \emph{vertical}
decomposition operation. It partitions the set of tables
into groups and a \emph{modified sub-query} is evaluated in every group.
For convenience, we will represent the instance $I$ as $R_1[T_1], \cdots, R_k[T_k]$,
where $T_i$ is the set of tuples in table $R_i$. Similarly, for a subset of tuples
$T_i' \subseteq T_i$, $R_i[T_i']$ will denote the instance of relation $R_i$
containing exactly the tuples in $T_i'$.
The algorithm \compsepquery\ is given in Algorithm~\ref{alg:comp_sep_query}.

\begin{algorithm*}[h!t!]
\caption{\em \compsepquery($Q$, 
$I = \angb{R_1[T_1], \cdots, R_k[T_k]},$ $G_C$, $G_T$, $\flag$)}
{\bf Input: Query $Q$, 
tables $R_1[T_1], \cdots, R_k[T_k]$,
co-table graph $G_C$, table-adjacency graph $G_T$, and a boolean parameter \flag\ which is true
if and only if row decomposition is performed at the current step.\\} 
{\bf Output: If successful, the unique read-once form $f^*$ of the expression for $Q(I)$} 

\begin{algorithmic}[1] \label{alg:comp_sep_query}
\IF { $k = 1$}
	\RETURN{$\sum_{x \in T_1} x$ with success. (\textit{/* all unused tuples are already removed */})}\label{step:base}
\ENDIF
\IF[\textit{/* Row decomposition */}]{\flag\ = \true}
	\STATE{-- Perform $\rowdecomp(\angb{R_1[T_1], \cdots, R_k[T_k]}, G_C)$.}
	\IF [\textit{/* \rowdecomp\ partitions every table and $G_C$ into $\ell \geq 2$ disjoint groups*/}] {row decomposition returns with success}
	\STATE{-- Let the groups returned be $\angb{\angb{T_1^j, \cdots, T_k^j}, G_{C,j}}$, $j \in [1, \ell]$.}
		\STATE{-- $\forall j \in [1, \ell]$, let $f_{j} = $ $\compsepquery(Q$, 
		$\angb{R_1[T_1^j], \cdots, R_k[T_k^j]}$, $G_{C,j}, G_T, \false)$.}
		\RETURN{ $f^* = f_1 + \cdots + f_\ell$ with success.}
	\ENDIF
\ELSE[\textit{/* Table decomposition */}] 
	\STATE{-- Perform $\tabledecomp(\angb{R_1[T_1], \cdots, R_k[T_k]}$, $Q$, 
	$G_T, G_C)$.}
	\IF[\textit{/* \tabledecomp\ partitions $I, G_C$ and $G_T$ into $\ell \geq 2$ disjoint groups, 
	$\sum_{j = 1}^\ell k_{j} = k$ */}]{table decomposition returns with success}
	\STATE{Let the groups returned be $\angb{\angb{R_{j, 1}, \cdots, R_{j, k_j}}, \widehat{Q_j}, G_{C,j}, G_{T, j}}$,
	$j \in [1, \ell]$.}
	\STATE{-- $\forall j \in [1, \ell]$, $f_{j} = $ $\compsepquery(\widehat{Q_j}$, 
		$\angb{R_1[T_1], \cdots, R_k[T_k]}$, $G_{C, j}, G_{T, j}, \true)$.}
	\RETURN{ $f^* = f_1 \cdot \ldots \cdot f_\ell$ with success.}
	\ENDIF
\ENDIF	
\IF[\textit{/* Current row or table decomposition is not successful and $k > 1$*/}]{the current operation is not successful}
	\RETURN{ with failure: ``$Q(I)$ is not read-once''.} 
\ENDIF	
\end{algorithmic}
\end{algorithm*}

\par

\medskip
\noindent
\textbf{Row Decomposition. ~}
 The row decomposition operation partitions the tuples variables
in \emph{every} table into $\ell$ disjoint groups. In addition, it decomposes the co-table
graph $G_C$ into $\ell \geq 2$ disjoint \emph{induced subgraphs}\footnote{A
subgraph $H$ of $G$ is an induced subgraph, if for any two vertices 
$u, v \in V(H)$, if $(u, v) \in E(G)$, then $(u, v) \in E(H)$.} corresponding to the
above groups. For every pair of distinct groups
$j, j'$, and for every pair of distinct tables $R_{i}, R_{i'}$, 
no tuple in group $j$ of $R_{i}$ ever joins with a tuple in group $j'$ of $R_{i'}$
(recall that the query does not have any self-join operation).
The procedure for row decomposition is given in Algorithm~\ref{alg:query_row_decomp}\footnote{It should be noted that
the row decomposition procedure  may be called on $R_{i_1}[T_{i_1}'], \cdots, R_{i_p}[T_{i_p}']$ and $G_C'$, 
where $R_{i_1}, \cdots, R_{i_p}$ is a subset of the relations from $R_1, \cdots, R_k$, 
$T_{i_1}', \cdots, T_{i_p}'$ are subsets of the respective set of tuples
$T_{i_1}, \cdots, T_{i_p}$, and $G_C'$ is the induced subgraph of $G_C$
on $T_{i_1}', \cdots, T_{i_p}'$.  For simplicity in notations, we use $R_{1}[T_{1}], \cdots, R_{k}[T_{k}]$. 
This holds for table decomposition as well.}.

\begin{algorithm}[h!t!]
\caption{{\em $\rowdecomp(\angb{R_{1}[T_{1}], \cdots, R_{k}[T_{k}]}, G_C')$}
}
{{\bf Input: Tables $R_{1}[T_{1}], \cdots, R_{k}[T_{k}]$, and induced subgraph 
$G_C'$ of $G_C$ on $\bigcup_{i=1}^k T_i$} }\\
{\bf Output: If successful, the partition of $G_C'$ and 
tuple variables of every input tables into $\ell \geq 2$ connected components: 
$\angb{\angb{T_{1,j}, \cdots, T_{k,j}}, G'_{C, j}}$, $j \in [1, \ell]$} 

\begin{algorithmic}[1] \label{alg:query_row_decomp}
\STATE{-- Run BFS or DFS to find the connected components in $G_C'$.}
\STATE{-- Let $\ell$ be the number of connected components.}
\IF[\textit{/* there is only one connected component */}]{$\ell = 1$}
	\RETURN{ with failure: ``Row decomposition is not possible''.}
\ELSE	
	\STATE{-- Let the tuples (vertices) of table $R_i$ in the $j$-th connected component $j$ of $G_C'$ be $T_{i,j}$}
	\STATE{-- Let the induced subgraph for connected component $j$ be $G_{C, j}$.}
	\RETURN{ $\angb{\angb{T_{1,1}, \cdots, T_{k,1}}, G'_{C, 1}}$, $\cdots$, 
		$\angb{\angb{T_{1,\ell}, \cdots, T_{k,\ell}}, G'_{C, \ell}}$ with success.}
\ENDIF	
\end{algorithmic}
\end{algorithm}



\medskip
\noindent
\textbf{Table Decomposition. ~}
On the other hand, the table decomposition operation partitions the set of tables $\mathbf{R} = R_1, \cdots, R_k$
into $\ell \geq 2$ disjoint groups $\mathbf{R_1}, \cdots, \mathbf{R_\ell}$.  
It also decomposes the table-adjacency graph 
$G_T$ and co-table graph $G_C$ into $\ell$ disjoint induced subgraphs $G_{T,1}, \cdots, G_{T, \ell}$,
and, $G_{C,1}, \cdots, G_{C,\ell}$ respectively
corresponding to the above groups. The groups are  selected in such a way that
all tuples in the tables in one group join with all tuples in the tables in another
group. This procedure also modifies the
sub-query to be evaluated on every group by making the subqueries of different
groups mutually independent by introducing free variables, i.e., they do not share any common variables
after a successful table decomposition.  Algorithm~\ref{alg:query_table_decomp}
describes the table decomposition operation.
Since the table decomposition procedure changes the input query $Q$ to $\widehat{Q} = \widehat{Q_1}, \cdots, \widehat{Q_\ell}$,
it is crucial to ensure that changing the query to be evaluated does not change
the answer to the final expression. This is shown in Lemma~\ref{lem:td_correct} in Appendix~\ref{app:proofsec4}.
\par
%

\begin{algorithm*}[h!t!]
\caption{\em $\tabledecomp(\angb{R_{1}[T_{1}], \cdots, R_{k}[T_{k}]}$, 
$\angb{Q() :- R_{1}(\mathbf{x_{1}}), \cdots, R_{k}(\mathbf{x_{k}})}$, $G_C', G_T')$} 
{\bf Input: Tables $R_{1}[T_{1}], \cdots, R_{k}[T_{k}]$ 
query $Q() :- R_{1}(\mathbf{x_{1}}), \cdots, R_{k}(\mathbf{x_{k}})$ 
induced subgraph $G_T'$ of $G_T$ on $\bigcup_{i=1}^k R_{i}$,
induced subgraph $G_C'$ of $G_C$ on $\bigcup_{i=1}^k T_{i}$}\\
{\bf Output: If successful, a partition of input tables, $G_T'$, $G_C'$ into $\ell$ groups,
 and an updated sub-query for every group} 

\begin{algorithmic}[1] \label{alg:query_table_decomp}
\FOR{all edges $e = (R_i, R_{j})$ in $G_T'$}
	\STATE{-- Annotate the edge $e$ with common variables $C_e$ in the vectors $\mathbf{x_i}$, $\mathbf{x_{j}}$.}
	\STATE{-- Mark the edge $e$ with a ``$+$'' if for every pair of tuple variables $x \in T_i$ and $y \in T_{j}$,
the edge $(x, y)$ exists in $G_C'$. Otherwise mark the edge with a ``$-$''.}\label{step:t_g_c}
\ENDFOR
\STATE{-- Run BFS or DFS to find the connected components in $G_T$ 
w.r.t ``$-$'' edges}
\STATE{-- Let $\ell$ be the number of connected components.}
\IF[\textit{/* there is only one connected component */}]{$\ell = 1$}
	\RETURN{ with ``Failure: Table decomposition is not possible''.}
\ELSE			
	\STATE{-- Let $G'_{T, 1}, \cdots, G'_{T, \ell}$ be the induced subgraphs of $\ell$ connected components of $G_T'$
	and $G'_{C,1}, \cdots, G'_{C, \ell}$ be the corresponding induced subgraph for $G_C'$.}
	\STATE{-- Let $\mathbf{R_p} = \angb{R_{p, 1}, \cdots, R_{p, k_p}}$ be the subset of tables
	in the $p$-th component of $G_T'$, $p \in [1, \ell]$.}
	\STATE{\textit{/* Compute a new query for every component */}}
	
	\FOR{every component $p$} \label{step:pivot}
     \FOR{every table $R_i$ in this component $p$}
				\STATE{-- Let $C_i = \bigcup_e C_e$ be the union of common variables $C_e$ over all edges $e$ from $R_i$
			to tables in different components of $G_{T'}$ (all such edges are marked with `+')}
				\STATE{-- For every common variable $z \in C_i$, generate a \emph{new (free)} variable $z^i$, and replace \emph{all}
			occurrences of $z$ in vector $\mathbf{x_i}$ by $z'$. Let $\widehat{\mathbf{x_i}}$ be the new vector.}
				\STATE{-- Change the query subgoal for $R_i$ from $R_i(\mathbf{x_i})$ to $R_i(\widehat{\mathbf{x_i}})$.}
		\ENDFOR
	\STATE{Let $\widehat{Q_p}() :- R_{p, 1}(\widehat{\mathbf{x_{p,1}}}), \cdots, R_{p, k_p}(\widehat{\mathbf{x_{p,k_p}}})$ 
	be the new query for component $p$.}

	\ENDFOR
	\RETURN{$\angb{\angb{R_{1, 1}, \cdots, R_{1, k_1}}, \widehat{Q_1}, G_{C, 1}, G_{T, 1}}$, $\cdots$, 
	$\angb{\angb{R_{\ell, 1}, \cdots, R_{\ell, k_{\ell}}}, \widehat{Q_\ell}, G_{C, \ell}, G_{T, \ell}}$ with success.}
\ENDIF	
\end{algorithmic}
\end{algorithm*}

The following lemma shows that if row-decomposition is successful, then table
decomposition cannot be successful and vice versa. However, both of them may be unsuccessful in case
the final expression is not read-once. The proof of the lemma is in Appendix~\ref{app:proofsec4}).

\begin{lemma}\label{lem:either_r_or_t}
At any step of the recursion, if row decomposition is successful then table decomposition is unsuccessful and
vice versa.
\end{lemma}


Therefore, in the top-most level of the recursive procedure, we can verify which operation
can be performed -- if both of them fail, then the final expression
is not read-once which follows from the correctness of our algorithm. 
If the top-most recursive
call performs a successful row decomposition
initially the algorithm $\compsepquery$ is called as 
\compsepquery($Q$, 
$\angb{R_1[T_1], \cdots, R_k[T_k]},$ $G_C$, $G_T$, $\true$).
The last boolean argument is \true\ if and only if row decomposition is performed at the 
current level of the recursion tree. If in the first step table decomposition 
is successful, then the value of the last boolean
variable in the initial call will be \false.

\par

\medskip
\noindent
\textbf{Correctness.~~}
The following two lemmas respectively show the soundness and completeness
of the algorithm \compsepquery\ (proofs are in Appendix~\ref{app:proofsec4}).
\begin{lemma}
\label{lemma:soundness}
\textbf{(Soundness)~~}If the algorithm returns with success, then the expression $f^*$ 
returned by the algorithm \compsepquery\
is equivalent to the expression $Q(I)$ generated by evaluation of query $Q$ on instance $I$. 
Further, the output expression $f^*$ is in read-once form.
\end{lemma}

\begin{lemma}
\label{lemma:completeness}
 \textbf{(Completeness)~~} If the expression $Q(I)$
 is read-once, then the algorithm \compsepquery\ returns the unique read-once form $f^*$ of the expression.
\end{lemma}

For completeness, it suffices to show that
if 
$Q(I)$ is read-once,
then the algorithm does not exit with error. Indeed, if the algorithm returns with success,
as showed in the soundness lemma, the algorithm returns an expression $f^*$ in read-once form
which is the unique read-once form of $Q(I)$ \cite{GolumbicMR06},\cite{CorneilLB81}.

\paragraph{Time Complexity.} 
Consider the recursion tree of the algorithm \compsepquery.
Lemma~\ref{lem:either_r_or_t} shows that at any level of the recursion tree,
either all recursive calls use the row decomposition procedure, 
or all recursive calls use the column decomposition procedure.
The time complexity of \compsepquery\ given in Theorem~\ref{thm:query_time}
is analyzed in the following steps.  
If $n' = $ the total number of input tuples at the current recursive
call and $m_C'$ = the number of edges in the induced subgraph of $G_C'$
on these $n'$ vertices, we show that row decomposition takes $O(m_C' + n')$ time
and, \emph{not considering the time needed to compute the modified queries $\widehat{Q_{j}}$}
(Step~\ref{step:pivot} in Algorithm~\ref{alg:query_table_decomp}), the table decompositions procedure takes 
$O(m_C' + n')$ time.
Then we consider the time needed to compute the
modified queries and show that these steps over \emph{all} recursive calls of the algorithm  
take $O(m_T \alpha \log \alpha)$ time in total, where $\alpha$
is the maximum size of a subgoal in the query $Q$.
Finally, we give a bound of $O(\min(k, \sqrt{n}))$ on 
the height of the recursive tree for the algorithm $\compsepquery$. 
However, note that at every step, for row or table decomposition, every tuple in $G_C'$
goes to exactly one of the recursive calls, and every edge in $G_{C}'$ goes to at most
one of the recursive calls. So for both row and table decomposition at every level
of the recursion tree the total time is $O(m_C + n)$. Combining all these observations,
the total time complexity of the algorithm is $O(m_T \alpha \log \alpha + (m_C+n) \min(k, \sqrt{n}))$
as stated in Theorem~\ref{thm:query_time}.
The details can be found in Appendix ~\ref{sec:app_time_query}). 

\medskip
\noindent

\paragraph{Example.~~} 
Here we illustrate our algorithm. 
Consider the query $Q$ and instance $I$ from Example~\ref{example1} in the introduction. 
The input query is $Q() :- R(x)S(x,y)T(y)$. 
In the first phase, the table-adjacency graph $G_T$ and the co-table graph $G_{C}$ are computed.
 These graphs are depicted in Fig.~\ref{fig:tadj} and Fig.~\ref{fig:cotable} respectively.
\par
Now we apply \compsepquery. There is a successful row decomposition at the top-most recursive call
that decomposes $G_{C}$ into the two subgraphs $G_{C,1},G_{C,2}$ shown in Fig.~\ref{fig:cotables}. 
So the final expression $f^*$ 
annotating the answer $Q(I)$ will be the \emph{sum} of the expressions $f_1, f_2$ 
annotating the answers of $Q$ applied to the relations corresponding to  $G_{C,1}$ and $G_{C,2}$ respectively.

\begin{figure}[t]
\begin{center}
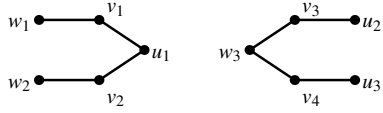
\caption{$G_{C,1}$ and $G_{C,2}$.}
\label{fig:cotables}
\vspace{-.3in}
\end{center}
\end{figure}

%
%
%

\eat{
\begin{figure}[t]
\begin{center}
\input{cotable1.pdf_t}
\caption{$G_{C,1}$.}
\label{fig:cotable1}
\end{center}
\end{figure}

\begin{figure}[t]
\begin{center}
\input{cotable2.pdf_t}
\caption{$G_{C,2}$.}
\label{fig:cotable2}
\end{center}
\end{figure}
}
The relations corresponding to $G_{C,1}$ are
$$
R =
\begin{array}{|c|c}
\cline{1-1}
a_1 & w_1 \\
\cline{1-1}
b_1 & w_2 \\
\cline{1-1}
\end{array}~~~~~
S =
\begin{array}{|cc|c}
\cline{1-2}
a_1 & c_1 & v_1 \\
\cline{1-2}
b_1 & c_1 & v_2 \\
\cline{1-2}
\end{array}~~~~~
T =
\begin{array}{|c|c}
\cline{1-1}
c_1 & u_1 \\
\cline{1-1}
\end{array}
$$
and, the relations corresponding to $G_{C,2}$ are 
$$
R =
\begin{array}{|c|c}
\cline{1-1}
a_2 & w_3 \\
\cline{1-1}
\end{array}~~~~~
S =
\begin{array}{|cc|c}
\cline{1-2}
a_2 & c_2 & v_3 \\
\cline{1-2}
a_2 & d_2 & v_4 \\
\cline{1-2}
\end{array}~~~~~
T =
\begin{array}{|c|c}
\cline{1-1}
c_2 & u_2 \\
\cline{1-1}
d_2 & u_3 \\
\cline{1-1}
\end{array}
$$

Now we focus on the first recursive call at the second level of recursion tree
with input co-table subgraph $G_{C, 1}$.
Note that the table-adjacency graph for this call is the same as $G_T$.
At this level the table decomposition procedure is invoked and the edges of the table-adjacency graph are marked with $+$ and $-$ signs, see Fig.~\ref{fig:tadjmark1}. 
In this figure the common variable set for $R, S$ on the edge $(R, S)$ is $\{x\}$, and for $S, T$ on the edge $(S, T)$
is $\{y\}$. Further, the edge $(S, T)$ is marked with a ``$+$'' because there are all possible edges
between the tuples in $S$ (in this case tuples $v_1, v_2$) and the tuples in $T$ (in this case $u_1$).
However, tuples in $R$ (here $w_1, w_2$) and tuples in $S$ (here $v_1, v_2$) do not have all possible
edges between them so the edge $(R, S)$ is marked with a ``$-$''.
\par
Table decomposition procedure performs a connected component decomposition 
using ``$-$''-edges, that 
decomposes $G_T$ in two components $\{R, S\}$ and $\{T\}$.
The subset $C$ of common variables collected from the ``$+$''-edges across different components
will be the variables on the single edge $(S, T)$, $C = \{y\}$.
This variable $y$ is replaced by new \emph{free variables} in all subgoals containing it, which are $S$ and $T$ in our case.
So the modified queries for disjoint components returned by the table decomposition procedure are $\widehat{Q_1}() :- R(x) S(x, y_1)$
and $\widehat{Q_2}() :- T(y_2)$. The input graph $G_{C,1}$ is decomposed further into $G_{C,1,1}$
and $G_{C,1,2}$, where $G_{C, 1, 1}$ will have the edges $(w_1, v_1)$ and $(w_2, v_2)$, whereas
$G_{C, 1, 2}$ will have no edges and a single vertex $u_1$.
Moreover, the expression $f_1$ is the product of $f_{11}$ and $f_{12}$ generated by these two queries
respectively. Since the number of tables for $\widehat{Q_2}$ is only one, and $T$ has a single tuple,
by the base step (Step~\ref{step:base}) of \compsepquery,
$f_{12} = u_1$. For expression $f_{11}$ from $\widehat{Q_1}$, now the graph $G_{C,1,1}$ can be decomposed using a row decomposition
to two connected components with single edges each ($(w_1, v_1)$ and $(w_2, v_2)$ respectively).
There will be recursive subcalls on these two components and each one of them will perform
a table decomposition (one tuple in every table, so the single edges in both calls will be marked with ``$+$'').
Hence $f_{11}$ will be evaluated to  $f_{11} = w_1v_1 + w_2v_2$.
So $f_1 = f_{11}\cdot f_{12} = (w_1v_1 + w_2v_2)u_1$.

\begin{figure}[t]
\begin{center}
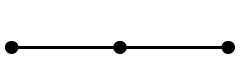
\caption{Marked table-adjacency graph for $R,S,T$.}
\label{fig:tadjmark1}
\end{center}
\end{figure}

\par

By a similar analysis as above, it can be shown that the same query $Q$ evaluated on the tables $R, S, T$
given in the above tables give $f_2 = w_3(v_3v_4+u_2u_3)$. So the
overall algorithm is successful and outputs the read-once form $f^* = f_1 + f_2$
$= (w_1v_1+w_2v_2)u_1+w_3(v_3u_2+v_4u_3)$. 

\section{Discussion of Time Complexity of 
Query-Answering
Algorithm}
\label{sec:complexity}
\scream{file complex-discussion.tex}

Putting together our results from Sections~\ref{sec:co-graph} 
and~\ref{sec:algo-query}, we propose the following algorithm
for answering boolean conjunctive queries without self-joins
on tuple-independent probabilistic databases.

\smallskip
\noindent
{\bf Phase 0} (Compute provenance DAG)\\
{\bf Input:} query $Q$, event table rep $I$\\
{\bf Output:} provenance DAG $H$\\
{\bf Complexity:} $O((\frac{n \rm{e}}{k})^k)$ \scream{see below}\\
\screamresolved{Nice: Choose ANY plan}

\smallskip
\noindent
{\bf Phase 1} (Compute co-table graph)\\
{\bf Input:} $H$, $Q$\\
{\bf Output:} table-adjacency graph $G_T$, co-table graph $G_C$\\
{\bf Complexity:} $O(n m_H + \beta_H m_{co} + k^2\alpha \log \alpha)$
(Thm.~\ref{thm:lca})\\
\screamresolved{\scream{See what can we do with $\beta$}}

\smallskip
\noindent
{\bf Phase 2} (Compute read-once form)\\
{\bf Input:} event table rep $I$, $Q$, $G_T$, $G_C$\\
{\bf Output:} read-once form $f^*$ or FAIL\\
{\bf Complexity:} $O(m_T\alpha \log \alpha + (m_C+n) \min(k, \sqrt{n}))$ 
(Thm.~\ref{thm:query_time})\\

\eat{
In this section first we 
discuss how the size of the input expression given as a DAG $H$ (represented by $m_H$)
relate to the other terms in this expression and then
the data complexity of our algorithm given a fixed query.
Then we compare the time complexity of our algorithm with that of previous algorithms
that computes the read-once form (if possible) of expressions generated by 
conjunctive queries without self-join.
}
\smallskip
\noindent
{\bf Size of the provenance DAG $H$.}~
Let $f$ be the boolean event expression generated by some query plan
for $Q$ on the database $I$.  The number of edges $m_H$ in the DAG $H$
represents the size of the expression $f$.  Since there are exactly
$k$ subgoals in the input query $Q$, one for every table, every
prime implicant of $f_{IDNF}$ will have exactly $k$
variables, so the size of $f_{IDNF}$ is at most ${n \choose k} \leq
(\frac{n \rm{e}}{k})^k$.  Further, the size of the input expression $f$ is
maximum when $f$ is already in IDNF.  So size of the
DAG $H$ is upper bounded by $m_H \leq (\frac{n \rm{e}}{k})^k$.  Again, the
``leaves'' of the DAG $H$ are exactly the $n$ variables in $f$. So
$m_H \geq n-1$, where the lower bound is achieved when the DAG $H$ is
a tree (every 
node in $H$ has a unique predecessor); in that
case $f$ must be read-once and $H$ is the unique read-once tree of
$f$.  \par Therefore  $n-1 \leq m_H \leq (\frac{n \rm{e}}{k})^k$.
Although the upper bound is quite high, it is our contention that
for practical query plans the size of the provenance graph is much
smaller than the size of the corresponding IDNF.\\  
\screamresolved{listed in my todolist :) dont have time to
include it to future work.
\scream{Good lesson for further research:
find plans that minimize the size of $H$. Of course this was done for safe
queries but here we see that it is good also for unsafe queries.}
}

\smallskip
\noindent
{\bf Data complexity.~} The complexity dichotomy of~\cite{DalviS04}
is for data complexity, i.e., the size of the query is bounded
by a constant. This means that our $k$ and $\alpha$ are $O(1)$.
Hence the time complexities of the Phase 1 and Phase 2 are
$O(n m_H + \beta_H m_{co})$ and
$O(m_C+n) \min(k, \sqrt{n}))$ respectively.  As discussed above,
$m_H$ is $\Omega(n)$ and $O((\frac{n \rm{e}}{k})^k)$.  
So one of these two terms may dominate the other
based on the relative values of $m_H, m_C$ and $m_{co}$ and of
$\beta_H$.  For example, when $m_H = \theta(n)$, $m_{co} = \theta(m_{C})
= \theta(n^2)$, and $\beta_H = O(1)$, the first phase takes $O(n^2)$
time, whereas the second phase may take $O(n^{\frac{5}{2}})$ time.
However, when $m_H = \Omega(n^{\frac{3}{2}})$, the first phase always
dominates.

In any case, 
we take the same position as~\cite{SenDG10} that
for unsafe~\cite{DalviS04}
queries the competition comes from the approach that does 
{\em not} try to detect whether the formulas are read-once and
instead uses probabilistic inference~\cite{KollerF09} which
is in general EXPTIME. In contrast, our algorithm runs in PTIME,
and works for a larger class of queries than the safe 
queries~\cite{DalviS04} (but of course, not on all instances).\\

\smallskip
\noindent
{\bf Comparisons with other algorithms.~} For these comparisons
we do not restrict ourselves to data complexity, instead taking the various
parameters of the problem into consideration.

First consider the {\bf general read-once detection algorithm}. This
consists of choosing some plan for the query, computing the answer boolean
event expression $f$, computing its IDNF, and then using the (so far, best)
algorithm~\cite{GolumbicMR06} to check if $f$ is read-once and if so to
compute its read-once form. The problem with this approach is that
the read-once check is indeed done in time a low polynomial, but
{\em in the size of $f_{IDNF}$}. For example, consider a boolean query
like the one in Example~\ref{example2}. 
This is a query
that admits a plan (the safe plan!) that would generate
the event expression $(x_1+y_1)\cdots(x_n+y_n)$ on an instance
in which each $R_i$ has two tuples $x_i$ and $y_i$. This is a read-once
expression easily detected by our algorithm, which avoids the computation
of the IDNF.

Next consider the {\bf cograph-help algorithm} that we have already mentioned
and justified in Section~\ref{sec:prelims}. 
This consists of our Phase 0 and a slightly 
modified Phase 1 that computes the co-occurrence graph $G_{co}$,
followed by checking if $G_{co}$ is a cograph
using one of the linear-time algorithms given 
in~\cite{CorneilPS85, HabibP05,BretscherCHP08} which also
outputs the read-once form if possible. Since Phase 0 and Phase 1
are common we only need to compare the last phases.
 
The co-graph recognition algorithms will all run in 
time $O(m_{co} + n)$. Our Phase 2 complexity is better than this
when $m_{C}\min(k, \sqrt{n}) = o(m_{co})$. 
Although in the worst case this algorithm performs at least as well
as our algorithm 
(since $m_C$ may be $\theta(m_{co})$),
(i) almost always the time required in first phases will 
dominate, so the asymptotic running time
of both these algorithms will be comparable, 
(ii) as we have shown earlier, the ratio $\frac{m_{co}}{m_C}$
can be as large as $\Omega(n^2)$, and the benefit of this could 
be significantly exploited by caching co-table graphs
computed for other queries (see discussions in Section~\ref{sec:conclusions}),
and (iii) these linear time
algorithms use complicated data structures, 
whereas we use simple graphs given as adjacency lists
and connectivity-based algorithms, 
so our algorithms are simpler to implement and may run faster in practice.

\par
Finally we compare our algorithm against that given in~\cite{SenDG10}.
Let us call it the {\bf lineage-tree algorithm} since 
they take the lineage tree of the result as input
as opposed to the provenance DAG as we do.
Although~\cite{SenDG10} does not give a complete running time 
analysis of the lineage tree algorithm, for the brief discussion
we have, we can make, to the best of our understanding,
the following observations.

Every join node in the lineage tree has two children, and every project 
node can have arbitrary number of children. 
When the recursive algorithm computes the read-once trees
of every child of a project node, every pair of such read-once trees are 
merged which may take $O(n^2k^2)$
time for every single merge (since the variables in the read-once trees 
to be merged are repeated).
Without counting the time to construct the lineage tree
this algorithm may take $O(Nn^2k^2)$ time in the worst case,
where $N$ is the number of nodes in the lineage tree.
\par
Since~\cite{SenDG10} does not discuss constructing the lineage
tree we will also ignore our Phase 0.  
We are left with comparing $N$ with $m_H$. 
It is easy to see that the number of edges in the provenance DAG $H$,
$m_H = \theta(N)$, where $N$ is the number of nodes in the lineage tree, when both
originate from the same query plan\footnote{If we ``unfold'' provenance DAG $H$
to create the lineage tree, the tree will have exactly $m_H$ edges,
and the number of nodes in the tree will be $N = m_H + 1$.} 
Since the lineage-tree algorithm takes $O(Nn^2k^2)$ time in the worst case,
and we use $O(n m_H + \beta_H m_{co} + k^2\alpha \log \alpha) + O((m_C + n)\min(k, \sqrt{n}))$
$= O(nN + \beta_H n^2 + k^2\alpha \log \alpha + n^{\frac{5}{2}})$.
The width $\beta_H$ of the DAG $H$ in the worst case can be the number of nodes in $H$.
So our algorithm \emph{always} gives an $O(k^2)$ improvement in time complexity
over the lineage-tree algorithm given in~\cite{SenDG10} 
whereas the benefit can often be more.



\section{Related Work}
\label{sec:related}

\scream{file related.tex}

The beautiful complexity dichotomy result of~\cite{DalviS04}
classifying conjunctive queries without self-joins 
on tuple-independent databases into ``safe'' and
``unsafe'' has spurred and intensified interest in probabilistic 
databases.  Some
papers have extended the class of safe relational
queries~\cite{DalviS07a,OlteanuH08,OlteanuH09,DalviSS10}.  Others have
addressed the question of efficient query answering for unsafe queries
on {\em some} probabilistic databases. This includes mixing the
intensional and extensional approaches, in effect finding subplans
that yield read-once subexpressions in the event
expressions~\cite{JhaOS10}.  The technique identifies ``offending''
tuples that violate functional dependencies on which finding safe
plans relies and deals with them intensionally. It is not clear that
this approach would find the read-once forms that our algorithm finds.
The OBDD-based approach in~\cite{OlteanuH08} works also for some
unsafe queries on some databases. The SPROUT secondary-storage
operator~\cite{OlteanuHK09} can handle efficiently some unsafe queries
on databases satisfying certain functional dependencies.

Exactly like us,~\cite{SenDG10} looks to decide efficiently when the
extensional approach is applicable given a conjunctive query without
self-joins and a tuple-independent database. We have made comparisons
between the two papers in various places, especially in
Section~\ref{sec:complexity}. Here we only add that that our
algorithm deals with different query plans uniformly, while the
lineage tree algorithm needs to do more work for non-deep plans.
The graph structures used in our approach
bear some resemblance to the graph-based synopses for relational
selectivity estimation in \cite{SpiegelP06}.

The read-once property has been studied for some time, albeit under various
names~\cite{Hayes75,Gurvich77,Valiant84,HellersteinK89,Gurvich91,KarchmerLNSW93,PeerP95}. 
It was shown~\cite{HellersteinK89} that if RP$\neq$NP then
read-once cannot be checked in PTIME for arbitrary monotone boolean
formulas, but for formulas in IDNF (as input) read-once can be checked in 
PTIME~\cite{GolumbicMR06}. Our result here sheds new light on another
class of formulas for which such an efficient check can be done.

%
%

\section{Conclusions and Future Work}
\label{sec:conclusions}

\scream{file conclusions.tex}

We have investigated the problem of efficiently deciding when a
conjunctive query {\em without self-joins} applied to a {\em
  tuple-independent} probabilistic database representation yields
result representations featuring {\em read-once} boolean event
expressions (and, of course, efficiently computing their read-once
forms when they exist). We have given a complete and simple to
implement algorithm of low polynomial data complexity for this
problem, and we have compared our results with those of other
approaches.

As explained in the introduction, the
results of this paper do not constitute complexity dichotomies.
However, there is some hope that the novel proof of completeness that
we give for our algorithm may be of help for complexity dichotomy
results in the space coordinated by the type of queries and the
type of databases we have studied.

Of independent interest may be that we have also implicitly performed a
study of an interesting class of monotone boolean formulas, those that
can be represented by the provenance graphs of conjunctive queries
without self-joins (characterizations of this class of formulas that
do not mention the query or the database can be easily given).  We
have shown that for this class of formulas the read-once property is
decidable in low PTIME (the problem for arbitrary formulas is unlikely
to be in PTIME, unless RP=NP). Along the way we
have also given an efficient algorithm
for computing the co-occurrence graph of such formulas (in all the other
papers we have examined, computing the co-occurrence graph entails
an excursion through computing a DNF; this, of course, may be the best
one can do for arbitrary formulas, if RP$\neq$NP).
It is likely that nicely tractable class of boolean formulas may 
occur in other database applications, to be discovered.

For further work one obvious direction is to extend our study to
larger classes of queries and probabilistic
databases~\cite{DalviS07a,DalviSS10}. Recall from the discussion
in the introduction however, that the class of queries considered
should not be able to generate arbitrary monotone boolean expressions.
Thus, SPJU queries are too much (but it seems that our approach might
be immediately useful in tackling unions of conjunctive queries without
self-joins, provided the plans do the unions last).

On the more practical side, work needs to be done to apply our
approach to non-boolean queries, i.e., they return actual tables.
Essentially, one would work with the provenance graph associated with
each table (initial, intermediate, and final) computing simultaneously
the co-table graphs of the event expressions on the graph's roots. It
is likely that these co-table graphs can be represented together, with
ensuing economy.

However we believe that the most practical impact would have the {\em
  caching of co-table graphs} at the level of the system, over batches
of queries on the same database, since the more expensive step in our
algorithm is almost always the computation of the co-table graph
(see discussion in Section~\ref{sec:complexity}).

\eat{
, the expensive
step in our algorithm is to compute the co-table graph $G_C$, 
since that introduces a factor $m_H = \Omega(n)$ in the total time complexity.
}

This would work as follows, for a fixed database $I$.  When a (let's
say boolean for simplicity) conjunctive query $Q_1$ is processed,
consider also the query $\bar{Q_1}$ which is obtained from $Q_1$ by
replacing each {\em occurrence} of constants with a distinct fresh
FO variable. Moreover if an FO variable $x$ occurs several
times in a subgoal $R_i(\mathbf{x}_i)$ of $Q$ but does {\em not}
occur in any of the other subgoals (i.e., $x$ causes selections
but not joins), replace also each occurrence of $x$ with a distinct
fresh FO variable. In other words, $Q_1$ is doing what $\bar{Q_1}$
is doing, but it first applies some selections on the various tables
of $I$. We can say that $\bar{Q_1}$ is the ``join pattern'' behind
$Q_1$. Next, compute the co-table graph for $\bar{Q_1}$ on $I$ and 
{\em cache} it together with $\bar{Q_1}$.
It is not hard to see that the co-table graph for $Q_1$
can be efficiently computed from that of $\bar{Q_1}$
by a ``clean-up'' of those parts related to tuples of $I$ 
that do not satisfy the select conditions of $Q_1$.

When another query $Q_2$ is processed, check if its join-pattern
$\bar{Q_2}$ {\em matches} any of the join-patterns previously cached
(if not, we further cache {\em its} join-pattern and co-table graph).
Let's say it matches $\bar{Q_1}$. Without defining precisely what
``matches'' means, its salient property is that the co-table graph of
$\bar{Q_2}$ can be efficiently obtained from that of $\bar{Q_2}$ by
another clean-up, just of edges, guided by the table-adjacency graph
of $\bar{Q_2}$ (which is the same as that of $Q_2$).  It can be shown
that these clean-up phases add only an $O(n\alpha)$ to the running
time.

There are two practical challenges in this approach.  The first one is
efficiently finding in the cache some join-pattern that matches that
of an incoming query. Storing the join-patterns together into some
clever data structure might help. The second one is storing largish
numbers of cached co-table graphs. Here we observe that they can all
be stored with the same set of nodes and each edge would have a list
of the co-table graph it which it appears. Even these lists can be
large, in fact the number of all possible joint-patterns is
exponential in the size of the schema. More ideas are needed and
ultimately the viability of this caching technique can only be
determined experimentally.\\

\medskip
\noindent
\paragraph{Acknowledgement.~~} We thank Amol Deshpande and Dan Suciu
for useful discussions.

\bibliographystyle{abbrv}
\bibliography{bib} 

\newpage
\appendix
\section{Generalization of Example 1}
\label{sec:appexample}

\begin{example} 
$$
R=
\begin{array}{|c|c}
\cline{1-1}
a_1 & x_1 \\
\cline{1-1}
b_1 & y_1 \\
\cline{1-1}
a_2 & x_2 \\
\cline{1-1}
\ldots & \dots\\
\cline{1-1}
a_{2n-1} & x_{2n-1} \\
\cline{1-1}
b_{2n-1} & y_{2n-1} \\
\cline{1-1}
a_{2n} & x_{2n} \\
\cline{1-1}
\end{array}~~~~~
T=
\begin{array}{|c|c}
\cline{1-1}
c_1 & u_1 \\
\cline{1-1}
c_2 & u_2 \\
\cline{1-1}
d_2 & v_2 \\
\cline{1-1}
\ldots & \dots\\
\cline{1-1}
c_{2n-1} & u_{2n-1} \\
\cline{1-1}
c_{2n} & u_{2n} \\
\cline{1-1}
d_{2n} & v_{2n} \\
\cline{1-1}
\end{array}
$$

$$
S=
\begin{array}{|cc|c}
\cline{1-2}
a_1 & c_1 & z_1 \\
\cline{1-2}
b_1 & c_1 & z_2 \\
\cline{1-2}
a_2 & c_2 & z_3 \\
\cline{1-2}
a_2 & d_2 & z_4 \\
\cline{1-2}
\ldots & \dots & \ldots\\
\cline{1-2}
a_{2n-1} & c_{2n-1} & z_{4n-3} \\
\cline{1-2}
b_{2n-1} & c_{2n-1} & z_{4n-2} \\
\cline{1-2}
a_{2n} & c_{2n} & z_{4n-1} \\
\cline{1-2}
a_{2n} & d_{2n} & z_{4n} \\
\cline{1-2}
\end{array}
$$

Answer event expression with standard plan:
\begin{equation}
x_1z_1u_1 + y_1z_2u_1  + \cdots + x_{2n}z_{4n-1}u_{2n} +
x_{2n}z_{4n}v_{2n}
\end{equation}
Equivalent to:
\begin{equation}
(x_1z_1 + y_1z_2)u_1 + \cdots + x_{2n}(z_{4n-1}u_{2n} + z_{4n}v_{2n})
\end{equation}
which is in read-once form and whose probability can be computed in time
$O(n)$.
\end{example}

\section{PTIME Probability Computation for Non-Read-Once
Expressions}\label{sec:discussions}
The following example shows that 
there exists a query $Q$ and a probabilistic database instance $D$ such that the 
expression $E$ for evaluation of query $Q$ 
on database $D$ is not read-once but still the probability of $E$ being true
can be computed in poly-time.
\begin{example}
The database $D$ has three tables $R(A), S(A,B), T(B)$.
Table $S$ has $n$ tuples. The tuple in the $2i-1$-th row
has tuple $(a_{i}, b_i)$ and the tuple in the $2i$-th row
has tuple $(a_{i+1},b_{i})$. We suppose that $S$ is a \emph{deterministic} relation, i.e. all the tuples is $S$ belong to $S$ with probability one and are annotated with true. If $n=2k$ then table $R$ has $k+1$ tuples, if $n=2k-1$ then $R$ has $k$ tuple. The tuple in row $j$ of $R$ is $a_j$ and is annotated by $x_{2j-1}$.
Table $T$ has $k$ tuples, where $n=2k-1$ or $n=2k$: the tuple in row $j$ is $b_j$ and is annotated by $x_{2j}$. 
It can be verified that the expression $E_n$ annotating the answer to the query $Q()=R(A),S(A,B),T(B)$ is 
$$
E_n=x_1x_2+x_2x_3+\ldots x_{n-1}x_n+x_nx_{n+1}.
$$
An example with $n = 2$ is given in Figure~\ref{fig:eg-poly-non-sep}.
It can be easily verified that for all $n > 1$, the expression $E_n$
is not read-once. 
\begin{figure}[htbp]
\begin{center}
$$
S=
\begin{array}{|c || c|}
\hline
a_1 & x_1 \\
\hline
a_2 & x_3 \\
\hline
\end{array}~~~~~
R=
\begin{array}{|c c || c|}
\hline
a_1 & b_1 & z_1=1_{\mathbb{B}} \\
\hline 
a_2 & b_1 & z_2=1_{\mathbb{B}} \\
\hline
a_2 & b_2 & z_3=1_{\mathbb{B}} \\
\hline
\end{array}~~~~~
T=
\begin{array}{|c || c|}
\hline
b_1 & x_2 \\
\hline
b_2 & x_4 \\
\hline
\end{array}
$$
\end{center}
\caption{Illustration with $n = 3$, $E_3= x_1x_2+x_2x_3+x_3x_4$.}  
\label{fig:eg-poly-non-sep}
\end{figure}
\par

Next we show that $P_n = P(E_n)$ can be computed in poly-time in $n$ by dynamic programming.
Note that $P_1$ can be computed in $O(1)$ time. 
Suppose for all $\ell < n$, $P_\ell$ is computed and stored in an array. 
Then 

\begin{eqnarray}
P_n&=&P(x_1x_2+\ldots+x_{n-2}x_{n-1}+x_{n-1}x_{n}+x_{n}x_{n+1})\nonumber\\
& = &P(x_1x_2+\ldots+x_{n-2}x_{n-1}+x_{n-1}x_{n})+P(x_nx_{n+1})\nonumber\\
& &- P(x_{n}x_{n+1}[x_1x_2+\ldots+x_{n-2}x_{n-1}+x_{n-1}x_{n}])\nonumber\\
& = & P_{n-1} + P(x_nx_{n+1}) \\
&  & - P(x_{n}x_{n+1}[x_1x_2+\ldots+x_{n-2}x_{n-1}+x_{n-1}x_{n}])\label{equn:1}
\end{eqnarray}

Observe that:
\begin{eqnarray}
 & & P(x_{n}x_{n+1}[x_1x_2+\ldots+x_{n-2}x_{n-1}+x_{n-1}x_{n}]) \nonumber \\
 &=&  P(x_{n+1})P(x_{n}[x_1x_2+\ldots+x_{n-2}x_{n-1}+x_{n-1}x_{n}]) \nonumber\\
 &= &P(x_{n+1})P(x_1x_2x_n+\ldots+x_{n-2}x_{n-1}x_n+x_{n-1}x_n) \mbox{ (idempotency) } \nonumber\\
 & = & P(x_{n+1})P(x_1x_2x_n+\ldots+x_{n-3}x_{n-2}x_n+x_{n-1}x_n) \mbox{ (absorption) } \nonumber\\
 & = & P(x_{n+1})P(x_n)P(x_1x_2+\ldots+x_{n-3}x_{n-2}+x_{n-1}) \nonumber\\
 & = & P(x_nx_{n+1})[P(x_1x_2+\ldots+x_{n-3}x_{n-2})+P(x_{n-1})] \nonumber\\
 & = & P(x_nx_{n+1})[P_{n-3} + P(x_{n-1})]\label{equn:2}
\end{eqnarray}

From (\ref{equn:1}) and (\ref{equn:2}),
\begin{eqnarray*}
P_n &= & P_{n-1} + P(x_nx_{n+1}) - P(x_nx_{n+1})[P_{n-3} + P(x_{n-1})]\\
& = & P_{n-1}+ P(x_nx_{n+1})[1 - P_{n-3} - P(x_{n-1})]
\end{eqnarray*} 

Since the variables $x_i$-s are independent, $P(x_n x_{n+1}) = P(x_n) P(x_{n+1})$,
and while computing $P_n$, $P_{n-1}$ and $P_{n-3}$ are already available.
Hence $P_n$ can be computed in linear time.

\end{example}

\section{Proofs from Section~3} 
\label{app:proof_sec3}
\noindent
\textbf{Proof of Lemma~\ref{lemma:corcompcotable}}.\\
\medskip
\noindent\textsc{Lemma~\ref{lemma:corcompcotable}.}
{\it Algorithm \compcotab\ adds an edge $(x, y)$ to $G_{C}$ if and only if $x, y$
together appear in some implicant in $f_{IDNF}$ and the tables containing $x, y$
are adjacent in $G_T$.}

\begin{proof}
Suppose two variables $x, y$ belong to the same implicant in $f_{IDNF}$, and 
their tables are adjacent in $G_T$.
Then by Lemma~\ref{lem:dnf-lca}, there is a $\cdot$-node $u \in \lca(x, y)$,
and $x \in \var(v_1), y \in \var(v_2)$ for two distinct successors $v_1, v_2$ of $u$.
When the algorithm processes the node $u$, if an edge between $x, y$ is not added
in a previous step, the edge will be added. This shows the completeness of algorithm
\compcotab.
\par
Now we show the soundness of the algorithm.
Consider two variables $x, y$ such that either the tables containing them are not adjacent in 
$G_T$ or they do not belong together in any of the implicants in $f_{IDNF}$.
If the tables containing $x, y$ are not adjacent in $G_T$, clearly, the algorithm
never adds an edge between them -- so let us consider the case when 
 $x, y$ do not belong to the same implicant in $f_{IDNF}$.
Then by Lemma~\ref{lem:dnf-lca}, there is no $\cdot$-node $u \in \lca(x, y)$.
\par
Consider any iteration of the algorithm and consider that a node $u$ 
is processed by the algorithm in this iteration. If $u$ is a $+$-node or 
if either $x \notin \var(u)$ or $y \notin \var(u)$,  
again no edge is added between $x, y$. So assume that, $u$ is a $\cdot$-node and 
$x, y \in \var(u)$.
Then $u$ is a common ancestor of $x$ and $y$. But since $u \notin \lca(x, y)$, by definition
of least common ancestor set, there is a
successor $v$ of $u$ such that $v$ is an ancestor of both $x, y$ and therefore, $x, y \in \var(v)$.
However, by Corollary~\ref{lem:unique_successor}, since $x$ or $y$ cannot belong
to two distinct successors of node $u$, node $v$ must be the unique successor of $u$ such that $x, y \in \var(v)$.
Since \compcotab\ only joins variables from two distinct children,
no edge will be added between $x$ and $y$
in $G_C$.
\end{proof}

\subsection{Time Complexity of \compcotab}\label{sec:app_time_lca} 
First we prove the following two lemmas bounding the number of times any given pair of variables 
$x, y$ are considered by the algorithm. The first lemma shows that 
 the variables $x, y$ are considered
by algorithm \compcotab\ to add an edge between them in $G_{co}$ only when they together appear in an implicant in $f_{IDNF}$,
i.e. only if the edge actually should exist
in $G_{co}$.
\begin{lemma}\label{lem:bound-no-edge}
Consider any two variables $x, y$ and a $\cdot$-node $u$.
If $x, y$ do not appear together in an implicant in $f_{IDNF}$, $x, y$
do not belong to the variable sets $\var(v_1), \var(v_2)$ for two distinct successors $v_1, v_2$ of $u$.
\end{lemma}
\begin{proof}
This easily follows from Lemma~\ref{lem:dnf-lca} which says that if $x, y$ do not appear together in an implicant in $f_{IDNF}$,
then there is no $\cdot$-node in $\lca(x, y)$. So for every $\cdot$-node $u$, either (i) one of $x$ and $y$
$\notin \var(u)$, or, (ii) there is a unique successor $v$ of $u$ which is a common ancestor of $x, y$, i.e. both $x, y \in \var(v)$
(uniqueness follows from Corollary~\ref{lem:unique_successor}). 
\end{proof}
The second lemma bounds the number of times a pair $x, y$ is considered by the algorithm
to add an edge between them. 
\begin{lemma}\label{lem:bound-edge}
Suppose $x, y \in \var(f)$ be such that they together appear in an implicant $f_{IDNF}$.
Then algorithm \compcotab\ considers $x, y$ in Step~\ref{step:check} to add an edge between them maximum $\beta_H$ times,
where $\beta_H$ is the width of the provenance DAG $H$.
\end{lemma}
\begin{proof}
Note that the check in Step~\ref{step:check} is performed only when the current node $u$ is a $\cdot$-node.
Consider any $\cdot$-node $u$. (i) if either $x$ or $y$
is not in $\var(u)$, clearly, $x, y$ are not checked in this step,
otherwise, (ii) if both $x, y \in \var(u)$,
and $x, y \in \var(v)$ for a unique child $v$ of $u$, then also $x, y$ are not checked at this step,
otherwise, (iii) if $u$ joins $x, y$, i.e., $x \in \var(v_1)$, $y \in \var(v_2)$ for two distinct
children $v_1, v_2$ of $u$, then only $x, y$ are considered by the algorithm in Step~\ref{step:check}.
(and after this node $u$ is processed, both $x, y$ appear in $\var(u)$).\par
However, since the query does not have any self-joins, the only time two variables $x, y$
appear in two distinct successors of a $\cdot$-node $u$ when the query plan
joins a subset of tables containing the table for $x$ with a subset of tables containing
the table for $y$. So the pair $x, y$ is multiplied at a unique layer of $H$,
and the total number of times they are multiplied cannot exceed the total number of nodes
in the layer which is at most the width $\beta_H$ of the DAG $H$. 
\end{proof}
Now we complete the running time analysis of algorithm \compcotab.
%

\begin{lemma}\label{lem:time-cotab}
Given the table-adjacency graph $G_T$ and input query plan $H$, algorithm \compcotab\ can be implemented in 
time $O(\beta_H m_{co} + n m_H)$ time, where $m_{co}$ is the number of edges in the co-occurrence graph,
$m_{H}$ is the number of edges in the DAG $H$, 
$\beta_H$ is the width of the DAG $H$ and $n = |\var(f)|$.
\end{lemma}


\begin{proof}
Initialization step can be done in $O(n)$ time.
The topological sort can be done in $O(m_H + |V(H)|)$ time by 
any standard algorithm.
\par
At every node $u \in V(H)$, to compute set $\var(u)$, the algorithm scans $O(d_{u})$ successors of $u$,
where $d_u = $ the outdegree of node $u$ in $H$. 
Although by Corollary~\ref{lem:unique_successor}, 
for every two distinct children $v_1, v_2$ of a $\cdot$-node $u$, $\var(v_1) \cap \var(v_2) = \phi$,
they may have some overlap when $u$ is a $+$-node, and here the algorithm incurs an $O(n m_H)$ cost total
as follows: (i) create an $n$-length boolean array for $u$ initialized to all zero, (ii) scan $\var(v)$
list of very successor $v$ of $u$, for a variable $x \in \var(v)$, if the entry for $x$ in the boolean array is false
mark it as true, (iii) finally scan the boolean array again to collect the variables marked as true
for variables in $\var(u)$. At every node $u \in V(H)$, the algorithm spends $O(n d_{u})$ time, 
where $d_u = $ the outdegree of node $u$ in $H$. Hence the total time across all nodes
= $\sum_{u \in V(H)} O(n d_{u})$ = $O(n m_H)$.
\par
Every check in Step~\ref{step:check}, i.e.,
whether an edge $(x, y)$ has already been added
and whether the tables containing $x, y$ are adjacent in $G_T$ can be done in $O(1)$ time using $O(n^2+k^2) = O(n^2)$
space. Further, by Lemma~\ref{lem:bound-no-edge} and \ref{lem:bound-edge}, 
the number of such checks performed is $O(\beta_H m_{co})$.
Since $\var(f) \subseteq V(H)$, and $H$ is connected, $n \leq |V(H)| \leq |E(H)|$.
Hence the total time complexity is $O(n m_H + \beta_H m_{co})$.
\end{proof}

We can now finish to prove Theorem~\ref{thm:lca}. As shown in Section~\ref{sec:comp_g_t}, computation of the table-adjacency graph $G_T$ takes $O(k^2 \alpha \log \alpha)$
time and this proves the second part of Theorem~\ref{thm:lca}.
The time complexity analysis in Lemma~\ref{lem:time-cotab} also holds when we modify \compcotab\ to compute the co-occurrence graph $G_{co}$
instead of the co-table graph $G_C$: the only change is that we do not check whether the tables containing $x, y$
are adjacent in $G_T$. Further, we do not need to precompute the graph $G_T$. This proves the first part 
and completes the proof of Theorem~\ref{thm:lca}.


%

\section{Proofs from Section~4} 
\label{app:proofsec4}
\noindent
\textbf{Modified query in Algorithm~\ref{alg:query_table_decomp}
evaluates the same expression:}\\
\noindent
\begin{lemma}\label{lem:td_correct}
Suppose $I = R_{i_1}[T_{i_1}'], \cdots, R_{i_p}[T_{i_p}']$ be the set of input tables to 
the table decomposition procedure \tabledecomp\ and let $Q() :- R_{i_1}(\mathbf{x_{i_1}}), \cdots, R_{i_p}(\mathbf{x_{i_p}'})$
be the input query. Then the expression $g$ generated by evaluating query $Q$ on $I$
is exactly the same as evaluating $\widehat{Q}$ on $I$, where $\widehat{Q'} = \widehat{Q_1}, \cdots, \widehat{Q_\ell}$
is the conjunction of modified queries $\widehat{Q_j}$ returned by the procedure \tabledecomp\
for groups $j = 1$ to $\ell$.
\end{lemma}

\begin{proof}
We prove that a set of $p$ tuple variables taken from $p$ tables satisfy the original input query $Q'$
if and only if they satisfy the modified query $\widehat{Q}$. 
Since the new query subgoals make some of the original variables \emph{free}, by replacing them 
with new variables, clearly, if a set of tuples satisfy the original query they also satisfy the 
modified query. So we prove that the modified query does not introduce any erroneous
collection of tuples in the final answer.
\par
Consider a set of tuple variables s.t. the corresponding tuples satisfy the modified query.
Let us partition these variables according to the $\ell$ groups of tables as computed
by procedure \tabledecomp. Consider component $j$ of the partition and any table $R_i$
in component $j$. Recall that $C_i$
is the set of all variables on the ``$+$''-edges having one end at the table $i$ in component $j$.
A ``$+$'' edge between table $R_i$ and $R_{i'}$ implies that the edges 
between every tuple in $R_i$ and every tuple in $R_j$ exist, which in turn implies that,
all tuples in $R_i$ and $R_{i'}$ must have the \emph{same} values of the attributes
corresponding to $C_e = \mathbf{x_i} \cap \mathbf{x_j}$. Then any set of $p$ tuples
taken from $p$ tables must have the same value of attributes corresponding to variables
in $C_e$. In other words, every variable $z \in C_i$ can be replaced by 
a new free variable $z^i$ in every table $R_i$ in component $j$  (note that $C_i \subseteq \mathbf{x_i})$
without changing the final solution.
\end{proof}

\medskip
\noindent
\textbf{Proof of Lemma~\ref{lem:either_r_or_t}}.\\
\medskip
\noindent
\textsc{Lemma~\ref{lem:either_r_or_t}.}
{\it At any step of the recursion, if row decomposition is successful then table decomposition is unsuccessful and
vice versa.}

\begin{proof}
Consider any step of the recursive procedure, where the input tables
are $R_{i_1}[T_{i_1}'], \cdots, R_{i_q}[T_{i_q}']$ ($\forall j, T_{i_j}' \subseteq T_{i_j}$),
input query is $Q'() :- R_{i_1}(\mathbf{x_{i_1}}), \cdots, R_{i_q}(\mathbf{x_{i_q}})$, and the induced subgraphs
of $G_C$ and $G_T$ on current sets of tuples and tables are $G_C'$ and $G_T'$ respectively.
\par
Suppose row decomposition is successful, i.e., it is possible to decompose 
the tuples in $G_C'$ into $\ell \geq 2$ connected components. Consider any two tables $R_i, R_j$
such that the edge $(R_i, R_j)$ exists in $G_T'$, and consider their sub-tables $R_i[T_i^1]$
and $R_j[T_j^2]$ taken from two different connected components in $G_C'$.  Consider two arbitrary tuples
$x \in T_i^1$ and $x' \in T_j^2$. Since $x$ and $x'$ belong to two different
connected components in $G_C'$, then there is no edge $(x, x')$ in $G_C'$.
Hence by Step~\ref{step:t_g_c} of the table decomposition procedure, this edge $(R_i, R_j)$
will be marked by ``$-$''. Since $(R_i, R_j)$ was an arbitrary edge in $G_T'$, all edges in $G_T'$
will be marked by ``$-$'' and there will be a unique component in $G_T'$ using ``$-$'' edges.
Therefore, the table decomposition procedure will be unsuccessful.
\par
Now suppose table decomposition is successful, i.e., $G_T'$ can be decomposed into $\ell \geq 2$
components using ``$-$'' edges. 
Note that wlog. we can assume that the initial table-adjacency graph $G_T$ is connected. 
Otherwise, we can run the algorithm on different components of $G_T$ and multipled the final expressions 
from different components at the end. 
Since the procedure \tabledecomp\ returns induced subgraphs for every connected components,
the input subgraph $G_T'$ is always a connected graph
at every step of the recursion. 
Now consider any two tables $R_i$ and $R_j$ from two different groups components such that $(R_i, R_j)$
edge exists in $G_T'$ (such a pair must exist since the graph $G_T'$ is connected).
Since this edge is between two components of a successful table decomposition procedure,
it must be marked with ``$+$''. This implies that for any tuple $x \in R_i$
and any tuple $x' \in R_j$, the edge $(x, x')$ exists in $G_C'$ (which follows from
Step~\ref{step:t_g_c} of this procedure). This in turn implies that row decomposition
must fail at this step since the tables $R_i, R_j$ cannot be decomposed into two disjoint
components and the graph $G_C'$ will be connected through these tuples. 
\end{proof}

\medskip
\noindent
\textbf{Proof of Lemma~\ref{lemma:soundness}}.\\
\medskip
\noindent\textsc{Lemma~\ref{lemma:soundness}.}
{\it \textbf{(Soundness)~~}If the algorithm returns with success, then the expression $f^*$ 
returned by the algorithm \compsepquery\
is equivalent to the expression $Q(I)$ generated by evaluation of query $Q$ on instance $I$. 
Further, the output expression $f^*$ is in read-once form.}

\begin{proof}
We prove the lemma by induction on $n$, where $n = \bigcup_{i = 1}^k |T_i|$.
The base case follows when $n = 1$. In this case there is only one tuple $x$, hence 
$k$ must be 1 as well, and therefore, the algorithm returns $x$ in Step~\ref{step:base}. Here the algorithm trivially returns with success
and outputs a read-once form. The output is also correct, since computation of co-table graph
ensures that there is no unused tuple in the tables, and the unique tuple $x$ 
is the answer to query $Q$ on database $I$.
\par
Suppose the induction hypothesis holds for all databases with number of tuples $\leq n - 1$
and consider a database with $n$ tuples. If $k = 1$, then irrespective of the query, all tuples
in table $R_1$ satisfies the query $Q() :- R_1(\mathbf{x_1})$ (again, there are no
unused tuples), and therefore the algorithm correctly returns
$\sum_{x \in T_1} x$ as the answer which is also in read-once form.
So let us consider the case when $k \geq 2$. 
\par
(1) Suppose the current recursive call successfully performs row decomposition, and $\ell \geq 2$ components
$\angb{R_1[T_1^1], \cdots, R_k[T_k^1]}$,
	$\cdots,$ $\angb{R_1[T_1^\ell], \cdots, R_k[T_k^\ell]}$ are returned.
	By the row decomposition algorithm , it follows that for $x \in T_i^j$
	and $x' \in T_{i'}^{j'}$, $x$ and $x'$ do not appear together in any monomial
	in the DNF equivalent for $Q(I)$. So
	the tuples which row decomposition puts in different 
	components do not join with each other and then the final answer of 
	the query is the union of the answers of the queries on the different 
	components. Then the final expression is the sum of the final expressions 
	corresponding to the different components. Since all components have $<n$ tuples 
	and the algorithm did not return with error in any of the recursive calls, by the 
	inductive hypothesis all the expressions returned by the recursive calls are 
	the correct expressions and are in read-once form. Moreover these 
	expressions clearly do not share variables --  they correspond to tuples from 
	different tables since the query does not have a self-join. We conclude that the final 
	expression computed by the algorithm is the correct one and is in read-once form.
	 \par
	(2) Otherwise, suppose the current step successfully performs table decomposition.
	Let $\ell \geq 2$ groups 
	$\mathbf{R_1}, \cdots, \mathbf{R_\ell}$ 
	are returned. Correctness of table decomposition procedure, i.e., correctness of the expression
	$f^* = f_1 \cdot \cdots \cdots f_\ell$, when all the recursive calls return successfully
	follows from Lemma~\ref{lem:td_correct} using the induction hypothesis
 (the algorithm multiplies the expressions returned by different groups which
 themselves are correct by the inductive hypothesis).
	Further, since all components have $<n$ tuples, and the algorithm did not return with error
	in any of the recursive calls, all expressions returned by the recursive calls
	are in read-once form. Since they do not share any common variable, the final output
	expression is also in read-once form.
\end{proof} 

\medskip

\noindent
\textbf{Proof of Lemma~\ref{lemma:completeness}}.\\
\medskip
\noindent\textsc{Lemma~\ref{lemma:completeness}.}
{\it \textbf{(Completeness)~~} If the expression $Q(I)$
 is read-once, then the algorithm \compsepquery\ returns the unique read-once form $f^*$ of the expression.}

\begin{proof}
Suppose the expression is read-once and consider the tree representation $T^*$ of the unique read-once
form $f^*$ of the expression ($T^*$ is in \emph{canonical form} and has alternate levels of $+$ and $\cdot$ nodes,
which implies that every node in $T^*$ must have at least two children.). 
We prove the lemma by induction on the height $h$ of tree $T^*$.
\par
First consider the base case. If $h = 1$, then the tree must have a single node for a single tuple 
variable $x$. Then $k$ must be 1 and the algorithm returns the correct answer. So consider $h \geq 2$. 
\par
(1) Consider the case when root of the tree is a $+$ node. If $h = 2$, since we do not allow
union operation, $k$ must be 1 and all the tuples must belong to the same table $R_1$. 
This is taken care of by Step~\ref{step:base} of \compsepquery. If $h > 2$, then $k$ must be $\geq 2$
and the answer to the join operation must be non-empty. Every child of the root node
corresponds to a set of monomials which will be generated by the equivalent DNF 
expression $f_{DNF}$ for the subtree rooted at that child. Note that no two variables in two different
children of the root node can belong to any monomial together since the tree $T^*$
is in read-once form. 
In other words, they do not share an edge in $G_C$. Hence the component formed 
by the set of variables at a child will not have any edge to the set of variables
at another child of the root node. This shows that all variables at different children of the root
node will belong to different components by the row decomposition procedure.
\par
Now we show that variables at different children of the root node are put  to different
components by the row decomposition procedure, which shows that
the row decomposition algorithm will divide the tuples \emph{exactly} the same was
as the root of $T^*$ divides tuples among its children.
Since $T^*$ is in canonical read-once form and has alternate levels of $+$ and $\cdot$ nodes, 
then row decomposition cannot be done within the same subtree of a $+$ node. So all variables
in a subtree must form a connected component.
Since the root has $\geq 2$
children, in this case we will have a successful row decomposition operation. 
By inductive hypothesis, since the subtrees rooted at the children of the root are 
all in read-once form, the recursive calls of the algorithm on the 
corresponding subtrees are successful. 
Hence the overall algorithm at the top-most level will be successful. 
\par
(2) Now consider the case when root of the tree is a $\cdot$ node. Note that the $\cdot$
operator can only appear as a result of join operation. If the root has $\ell' \geq 2$
children $c_1, \cdots, c_{\ell'}$, then every  tuple $x$ in the subtree at $c_{j}$
joins with every tuple $y$ in the subtree at $c_{j'}$ for every pair $1 \leq j\neq j' \leq \ell'$.
Moreover since the query does not have a self join, $x, y$ must belong to two different tables,
which implies that there is an edge $(x, y)$ in $G_C$ between every pair of tuples
$x, y$ from subtrees at $c_j, c_{j'}$ respectively.
Again, since we do not allow self-join, and $T^*$ is in read-once form, the tuples in the subtrees
at $c_j, c_{j'}$ must  belong to different tables if $j \neq j'$. In other words,
the tables $R_1, \cdots, R_k$ are partitioned into $\ell'$ disjoint groups 
$\mathbf{R_1'}, \cdots, \mathbf{R_{\ell'}'}$. 
\par
Next we argue that $\ell = \ell'$ and 
the partition returned by
the table decomposition procedure $\mathbf{R_1}, \cdots, \mathbf{R_\ell}$ is identical
to $\mathbf{R_1'}, \cdots, \mathbf{R_{\ell'}'}$ upto a permutation of indices. 
Consider any pair $\mathbf{R_j'}$ and $\mathbf{R_{j'}'}$. Since the tuple variables in these two groups are 
connected by a $\cdot$ operator, all tuples in all tables in $\mathbf{R_j'}$ join with
all tuples in all tables in $\mathbf{R_{j'}'}$. In other words, for any pair of tuples 
$x, x'$ from $R_{i_1} \in \mathbf{R_{j}'}$ and $R_{i_2} \in \mathbf{R_{j'}'}$, 
there is an edge $(x, x')$ in co-occurrence graph $G_{co}$.
Hence if there is a common subset of join attributes between $R_{i_1}$ and $R_{i_2}$,
i.e. the edge $(R_{i_1}, R_{i_2})$ exists in $G_T$, it will be marked by a ``$+$''
(all possible edges between tuples will exist in the co-table graph $G_T$).
So the table decomposition procedure will put $\mathbf{R_j'}$ and $\mathbf{R_{j'}'}$ in two different components.
This shows that $\ell \geq \ell'$.
However, since $T^*$ is in read-once form and has alternate levels of $+$ and $\cdot$ nodes,
no $\mathbf{R_j'}$ can be decomposed further using join operation (i.e. using ``$+$'' marked edges
by the table decomposition procedure); therefore, $\ell' = \ell$.
Hence our table
decomposition operations exactly outputs the groups $\mathbf{R_1'}, \cdots, \mathbf{R_{\ell'}'}$.
By the inductive hypothesis the algorithm returns with success in all recursive calls,
and since $\ell' = \ell \geq 2$, the table decomposition returns with success. So
the algorithm returns with success.
\end{proof}


\subsection{Time Complexity of \compsepquery}\label{sec:app_time_query}

Here we discuss the time complexity of algorithm \compsepquery\ in detail 
and show that algorithm \compsepquery\
runs in time $O(m_T \alpha \log \alpha + (m_C + n) \min(k, \sqrt{n}))$. 
We divide the time complexity computation
in two parts: (i) total time required to compute the modified queries across
\emph{all} table decomposition steps performed by the algorithm (this will
give $O(m_T \alpha \log \alpha)$ time) and (ii) total time required
for all other steps: here we will \emph{ignore} the time complexity 
for the modified query computation step and will get a bound of $(m_C + n) \min(k, \sqrt{n}))$.
First we bound the time complexity of individual row decomposition and table decomposition steps.
\begin{lemma}\label{lem:row_decomp_poly}
The row decomposition procedure as given in Algorithm~\ref{alg:query_row_decomp} runs in 
time $O(m_C' + n')$, where $n' = \sum_{j = 1}^q |T_{i_j}|$ = the total number of input 
tuples to the procedure, and $m_C'$ = the number of edges in the induced subgraph of $G_C$
on these $n'$ tuples.
\end{lemma}
\begin{proof}
The row decomposition procedure only runs a connectivity algorithm like BFS/DFS to
compute the connected components. Then it collects and returns the tuples and computes the induced
subgraphs in these components. All these can be done in linear time in the size of the input graph
which is $O(m_C' + n')$.
\end{proof}
Next we show that the table decomposition can be executed in time $O(m_C' + n')$ as well.

\begin{lemma}\label{lem:table_decomp_poly}
The table decomposition procedure as given in Algorithm~\ref{alg:query_table_decomp} runs in 
time $O(m_C' + n')$, ignoring the time required to compute the modified queries $\widehat{Q_j}$
where $n' = \sum_{j = 1}^q |T_{i_j}'|$ = the total number of input 
tuples to the procedure, and $m_C'$ = the number of edges in the induced subgraph of $G_C$
on these $n'$ tuples.
\end{lemma}


\begin{proof}
Step~\ref{step:t_g_c} in the table decomposition procedure marks edges in $G_T'$
using $G_C'$. Let us assume that $G_C'$ has been represented in a standard adjacency list.
Consider a table $R_{j}[T_{j}']$, where $T_{j}' \subseteq T_{j}$ and let $d$ be the degree of
$R_{j}$ in $G_T'$.
Now a linear scan over the edges in $G_C'$ can partition the edges $e$ from a tuple $x \in T_{j}'$
in table $R_{j}$ into $E_1, \cdots, E_{d}$, where $E_q$ ($q \in [1, d]$) contains all edges from $x$ to tuples $x'$,
belonging to the $q$-th neighbor of $R_j$. A second linear scan on
these grouped adjacency lists computed in the previous step is sufficient to mark every edge in $G_T'$
with a ``$+$'' or a ``$-$'':  for every neighbor $q$ of $R_j$, say $R_{j'}$, for every tuple $x$ in $T_j'$, 
scan the $q$-th group in adjacency list to check if $x$ has edges with all tuples in $R_{j'}[T_{j'}']$. 
If yes, then all tuples in $R_{j'}$ also have edges to all tuples in $R_j$, and the edge $(R_j, R_{j'})$
is marked with a ``$+$''. Otherwise, the edge is marked with a ``$-$''. Hence the above two steps
take $O(m_C' + n' + m_T' + k')$ time, where $k'$ and
$m_t'$ are the number of vertices (number of input tables) and edges in the subgraph $G_T'$.
\par
Finally returning the induced subgraphs of $G_T'$ for the connected components and decomposition of the tuples
takes $O(m_C' + n' + m_T' + k')$ time. Since $n' \geq k'$ and $m_C' \geq m_T'$, 
not considering the time needed to recompute the queries,
step, the total time complexity is bounded by $O(m_C' + n')$.
\end{proof}

%
The next lemma bounds the total time required to compute the modified queries over all
calls to the recursive algorithm.
\begin{lemma}\label{lem:time_comp_S}
The modified queries $\widehat{Q_j}$ over all steps can be computed in time
$O(m_T\alpha \log \alpha)$, where $\alpha$ is the maximum size of a subgoal.
\end{lemma}
\begin{proof}
We will use a simple charging argument to prove this lemma.
For an edge $e = (R_i, R_j)$ in $G_T$, the common variable set $C_e = \mathbf{x_i} \cap \mathbf{x_j}$\footnote{We
abuse the notation and consider the \emph{sets} corresponding
to \emph{vectors} $\mathbf{x_i}, \mathbf{x_j}$ to compute $C_e$}
can be computed by (i) first sorting the variables in $\mathbf{x_i}, \mathbf{x_j}$
in some fixed order, and then (ii) doing a linear scan on these sorted lists to compute the common variables.
Here we  to compute the set $C_e$. Hence this step takes $O(\alpha \log \alpha)$ time.
Alternatively, we can use a hash table
to store the variables in $\mathbf{x_i}$, and then by a single scan of variables in $\mathbf{x_j}$
and using this hash table we can compute the common attribute set $C_e$ in $O(\alpha)$ expected time.
When $C_e$ has been computed in a fixed sorted order for every edge $e$ incident on $R_i$ to a different component,
the lists $C_e$-s can be repeatedly merged to compute the variables set $C_i = \bigcup_e C_e$
in $O(d_i \alpha)$ time (note that even after merging any number of $C_e$ sets, the individual
lists length are bounded by the subgoal size of $R_i$ which is bounded by $\alpha$).  
However, instead of considering the total time $O(d_i \alpha)$
for the node $R_i$ in $G_T$, we will \emph{charge} every such edge
 $e = (R_i, R_j)$ in $G_T$ for this merging procedure an amount of $O(\alpha)$.
 So every edge $e$ from $R_i$ to an $R_j$ in different component gets a charge of $O(\alpha \log \alpha)$.
\par
Suppose we charge the outgoing edges $(R_i, R_j)$ from $R_i$ to different components
by a fixed cost of $P$, $P = O(\alpha)$ in the above process.
From the table decomposition procedure it follows that,
the common join attributes are computed, and the query is updated, only when the edge $(R_i, R_j)$
belongs to the \emph{cut} between two connected components formed by the ``$-$'' edges.
These edges then get \emph{deleted} by the table decomposition procedure: 
all the following recursive calls consider the edges \emph{inside} these 
connected components and the edges \emph{between} two connected components are never considered later.
So each edge in the graph $G_T$ can be \emph{charged} at most once for
computation of common join attributes and this gives $O(m_C)\alpha \log \alpha$ as 
the total time required for this process. 
\par
Finally,
the variables in $\mathbf{x_i}$ can also be replaced by new variables using the sorted list for $C_i$
in $O(\alpha)$ time, so the total time needed is $O(m_C\alpha \log \alpha + n \alpha) = O(m_C\alpha \log \alpha)$
(since we assumed $G_T$ for the query $Q$ is connected without loss of generality). 
\end{proof}


Now we show that the depth of the recursion tree is $O(\min(k, \sqrt{n}))$
and in every level of the tree, the total time required is at most $O(m_C + n)$.
Let us consider the recursion tree of the algorithm \compsepquery\ and wlog.
assume that the top-most level performs a row decomposition. Since the size of the 
table-adjacency subgraph $G_T'$ is always dominated by the co-table subgraph $G_C'$
at any recursive call of the algorithm, we express the time complexity of the algorithm
with $k$ tables, and, $n$ tuples and $m$ edges in the subgraph $G_C'$ as $T_1(n, m, k)$, 
where the top-most operation is a row decomposition.
Further, every component after row decomposition has exactly $k$
tables and therefore must have at least $k$ tuples, because,
we assumed wlog. that initial table adjacency graph
$G_T$ is connected and there is no unused tuples in the tables.
Similarly, $T_2(n, m, k)$ denotes the time complexity when the top-most operation is 
a table decomposition operation.
Note that at every step, for row decomposition, every tuple and every edge in $G_C'$
goes to exactly one of the recursive calls
of the algorithm; however, the number of tables $k$ remains unchanged. 
On the other hand, for table decomposition operation, every tuple goes to exactly one
recursive call, every edge goes to at most one such calls (edges between connected components
are discarded), and every table goes to exactly one call.
Recall that the row and table decomposition alternates at every step, and the time
required for both steps is $O(m+n)$ (not considering computation of modified queries at every table
decomposition steps)
so we have the following recursive formula for $T_1(n, m, k)$ and $T_2(n, m, k)$.
\begin{eqnarray*}
T_1(n, m, k) & = & O(m + n) + \sum_{j = 1}^\ell T_2(n_{j}, m_{j}, k)\\
& & \quad \quad \quad \quad \text{where }\sum_{j = 1}^\ell n_j = n, \sum_{j = 1}^{\ell} m_j = m, n_j \geq k \forall j \\
T_2(n, m, k) & = & O(m + n) + \sum_{j = 1}^\ell T_1(n_{j}, m_{j}, k_{j})\\
& & \quad \quad \quad \quad \text{where }\sum_{j = 1}^\ell n_j = n, \sum_{j = 1}^\ell m_j \leq m, \sum_{j = 1}^\ell k_j = k 
\end{eqnarray*}
where $n_j, m_j$ and $k_j$ are the total number of tuples and edges in $G_C'$,
and the number of tables for the $j$-th recursive call (for row decomposition, $k_j = k$).
For the base case, we have $T_2(n_{j}, m_{j}, 1) = O(n_j)$ -- for $k = 1$, to compute the the read once form, 
$O(n_j)$ time is needed; also in this case $m_j = 0$
(a row decomposition cannot be a leaf in the recursion tree for a successful completion
of the algorithm). Moreover, it is important to note that for a successful row or table decomposition, $\ell \geq 2$.
\par
If we draw the recursion tree for $T_1(n, m_C, k)$ (assuming the top-most operation is a row-decomposition
operation), at every level of the tree we pay cost at most $O(m_C + n)$. This is because the tuples and edges go to 
at most one of the recursive calls and $k$ does not play a role at any node of the recursion tree 
(and is absorbed by the term $O(m_C + n)$). 
\par
Now we give a bound on the height of the recursion tree.

\begin{lemma}
The height of the recursion tree is upper bounded by $O(\min(k, \sqrt{n}))$.
\end{lemma}

\begin{proof} 
Every internal node has at least two children and there are at most $k$ leaves (we return
from a path in the recursion tree when $k$ becomes 1). Therefore, there are $O(k)$ 
nodes in the tree and the height of the tree is bounded by $O(k)$ (note that
both the number of nodes and the height may be $\Theta(k)$ when the tree is not balanced). 
\par
Next we show that the height of the recursion tree is also bounded by $4\sqrt{n}$.
The recursion tree has alternate layers of table and row
decomposition. We focus on only the table decomposition layers, the
height of the tree will be at most twice the number of these layers.
Now consider any arbitrary path $P$ in the recursion tree from the root
to a leaf
where the number of table decompositions on $P$ is $h$. Suppose that in the calls $T_2(n, m, k)$, the values of $n$ and $k$
along this path (for the table decomposition layers) are $(n_0, k_0), (n_1, k_1),\ldots, (n_h, k_h)$,
where $k_0 = k$ and $n_0 \leq n$ (if the top-most level has a table-decomposition operation,
then $n_0 = n$). We show that $h \leq 2\sqrt{n}$.
\par
Let's assume the contradiction that $h > 2\sqrt{n}$ and
let's look at the first $p = 2\sqrt{n}$ levels along the path $P$. 
If at any $j$-th layer, $j \in [1, p]$, 
$k_j \leq 2\sqrt{n} - j$, then the number of table decomposition steps along $P$
is at most $2\sqrt{n}$: every node in the recursion tree has at least two children, so the value 
 of $k$ decreases by at least 1.
 The number of table-decomposition layers after the $j$-th node is at most $k_j$,
 and the number of table-decomposition layers before the $j$-th node is exactly $j$.
 Therefore, the total number of table-decomposition layers is $\leq 2\sqrt{n}$). 
 \par
 Otherwise, for all $j \in [1, p]$, $k_j > 2\sqrt{n} - j$. 
 Note that $n_j \leq n_{j-1} - k_{j-1}$: there is a row decomposition step between two table decompositions,
 and every component in the $j$-th row decomposition step will have at least $k_j$ nodes. 
 If this is the case, we show that $n_p < 0$.
 However,
 \begin{eqnarray*}
 n_p & \leq & n_{p-1} - k_{p-1}\\
 		 & \leq & n_{p-2} - k_{p-2} - k_{p-1}\\
 		 & \vdots & \\
 		 & \leq & n_0 - \sum_{j = 0}^{p-1} k_j\\
 		 & \leq & n - \sum_{j = 0}^{p-1} k_j\\
 		 & \leq & n - \sum_{j = 0}^{2\sqrt{n}-1} (2\sqrt{n} - j)\\
 		 & = & n - \sum_{j = 1}^{2\sqrt{n}} j\\
 		 & = & n - \frac{2\sqrt{n}(2\sqrt{n}+1)}{2}\\
 		 & = & n - 2n - \sqrt{n} \\ 
 		 & < & 0
 \end{eqnarray*}
 which is a contradiction since $n_p$ is the number of nodes at a recursive call
 and cannot be negative. This shows that along any path from root to leaves, the number of table decomposition
 layers is bounded by $2\sqrt{n}$ which in turn shows that the height of the tree is bounded by $4\sqrt{n}$.
%
%
\end{proof}
%
%

Since total time needed at every step of the recursion tree is $O(m_C + n)$, 
we have the following corollary,  
\begin{corollary}\label{cor:algo_query_poly}
Not considering the time complexity to compute the modified queries by the table decomposition procedure,
the algorithm \compsepquery\ runs in time $O((m_C + n)\min(k, \sqrt{n}))$. 
\end{corollary}

The above corollary together with Lemma~\ref{lem:time_comp_S}
(which says that to compute the modified queries $O(m_T \alpha \log \alpha)$ time suffices)
shows that \compsepquery\ runs in time $O(m_T \alpha \log \alpha + (m_C + n) \min(k, \sqrt{n}))$.



\end{document}